\documentclass[11pt]{article}
\usepackage{amsmath}
\usepackage{amsfonts}
\usepackage{amssymb}
\usepackage{amsthm}
\usepackage{fullpage}
\usepackage{graphicx}
\usepackage{color}
\def\be{\begin{equation}}
\def\ee{\end{equation}}
\def\bea{\begin{eqnarray}}
\def\eea{\end{eqnarray}}
\def\bma{\begin{mathletters}}
\def\ema{\end{mathletters}}

\def\0{\overline{0}}
\def\tr{\mbox{tr}}

\def\q0{\underline{0}}

\def\Z{{\mathbb Z}}

\def\id{{\mathbb I}}

\def\R{\mathbb{R}}
\def\N{\mathbb{N}}

\def\tr{\mbox{tr}}
\def\one{\leavevmode\hbox{\small1\normalsize\kern-.33em1}}

\def\bra#1{\langle#1|} \def\ket#1{|#1\rangle}
\def\braket#1#2{\langle#1|#2\rangle}

\def\proj#1{\ket{#1}\!\bra{#1}}

\def\G{{\cal G}}
\newtheorem{theo}{Theorem}

\newtheorem{lemma}[theo]{Lemma}
\newtheorem{prop}[theo]{Proposition}

\def\id{{\mathbb I}}
\def\tr{\mbox{tr}}

\begin{document}
\title{Testing Microscopic Discretization}
\author{Miguel Navascu\'es$^1$, David P\'erez-Garc\'ia$^2$ and Ignacio Villanueva$^2$\\$^1$School of Physics, University of Bristol, BS8 1TL, U.K.\\$^2$Facultad de Ciencias Matem\'aticas, Universidad Complutense de Madrid, 28040, Spain}

\date{}
\maketitle

\begin{abstract}
What can we say about the spectra of a collection of microscopic variables when only their coarse-grained sums are experimentally accessible? In this paper, using the tools and methodology from the study of quantum nonlocality, we develop a mathematical theory of the macroscopic fluctuations generated by ensembles of independent microscopic discrete systems. We provide algorithms to decide which multivariate gaussian distributions can be approximated by sums of finitely-valued random vectors. We study non-trivial cases where the microscopic variables have an unbounded range, as well as asymptotic scenarios with infinitely many macroscopic variables. From a foundational point of view, our results imply that bipartite gaussian states of light cannot be understood as beams of independent $d$-dimensional particle pairs. It is also shown that the classical description of certain macroscopic optical experiments, as opposed to the quantum one, requires variables with infinite cardinality spectra.
\end{abstract}

\section{Introduction}
The central limit theorem is one of the most celebrated results in statistics and applied probability. Originally formulated by Laplace, it states that the sum of many independent variables tends to a gaussian distribution in the asymptotic limit. The central limit theorem is used to estimate error bars in all sorts of physical experiments. It also justifies why the probability density of magnitudes as diverse as the IQ of a large sample of individuals or the noise in a radio signal approaches the so-called `Bell curve', by postulating that such quantities are the sum of many independent (microscopic) contributions\footnote{In the case of the IQ test, the scores of each question}.

In this work we somehow reverse this reasoning and consider the problem of extracting  information about the `micoscropic variables' provided that we know the experimental macroscopic data.
Specifically, we are interested in how different constraints over discrete microscopic models manifest in the macroscopic limit.

The reason why we restrict our study to discrete models is two-fold: on one hand, from a practical point of view, discrete systems adopting a finite (small) set of possible values are preferred in computational modeling, due to the modest memory resources required to store and manipulate them. If the internal mechanisms of a given -say- biological process at the cellular level are unknown but there is plenty of data available about the macroscopic behavior of a group of cells, it is thus advisable to search for discrete models for individual cell behavior. On the other hand, from a foundational point of view, there exist some physical quantities, like time and space, which, although traditionally regarded as continuous, could actually be discrete (see the causal set approach to Quantum Gravity  \cite{causal}). It is thus interesting to know if such a discretization leaves a signature at the macroscopic level.

In this paper we will perform a thorough analysis of this mathematical problem. We will show that, if the spectrum of the microscopic variables is finite and fixed a priori, then there always exist macroscopic bivariate distributions which cannot be approximated with such microscopic models. In fact, deciding which gaussian distributions are approximable or not can be cast as a linear program. We will also arrive at the surprising conclusion that certain spectra of infinite cardinality do \emph{not} allow to reproduce all bivariate covariance matrices. If the spectrum of the microscopic systems is free but finite, then several $n$-tuples of $n+1$-valued variables are needed to recover all $n$-variate gaussian distributions, and the problem of deciding if a given gaussian distribution is generated by $d$-valued microscopic systems can be solved via semidefinite programming.

Based on the ubiquitous applicability of the central limit theorem,  we expect our results to be applicable in a large variety of scientific situations where one wants to infer microscopic information from macroscopic data, such as molecular biology, genetics, neuroscience or social sciences. We will mainly illustrate this potential in the context of quantum mechanics, but we will also give some new results about Brownian motion in the plane and provide a operational meaning to the problem QPRATIO, recently introduced in complexity theory \cite{Bhaskara}.

Based on the mathematical results we give in the paper, the main application we can infer within quantum mechanics are ways of proving that a simple quantum experiment cannot have a simple classical explanation, where we measure the complexity of an explanation by the size of the underlying discrete sample space\footnote{Namely, the set of possible detector responses.}. This type of measure of complexity in the quantum case has become important in last years by (1) the realm of quantum information theory, where Hilbert spaces are finite dimensional and (2) the need to have small Hilbert space dimension in some security proofs of quantum key distribution protocols, like the original BB84.

Indeed, examples of how to get  experiments with a simple quantum explanation but a complex classical one have appeared recently in the quantum information literature, albeit in a completely different setup: snapshots of a Markovian evolution at different times \cite{Wolf}.
 
Along these lines, using the techniques developed along this paper, we will show that no finite classical model can account for the intensities observed when both sides of a bipartite state consisting of many copies of the maximally entangled state are subject to extensive equatorial spin measurements. Since each of the microscopic quantum variables involved can adopt one of two values (up or down) in each instance of the experiment, the quantum description can thus be considered infinitely less complex than its classical counterpart. This result has a direct implication in the field of Foundations of Quantum Mechanics: in \cite{mac_loc} it was shown that any Bell-type quantum experiment admits a classical explanation when brought to the macroscopic scale (this phenomenon was named \emph{macroscopic locality}). However, the particular classical model found in \cite{mac_loc} seemed excessively convoluted: whereas the quantum variables could adopt a finite number of values, the spectrum of the classical ones was continuous, ranging from $-\infty$ to $+\infty$. Now we know that one indeed needs such complexity for the classical model.

One can also obtain easily with the results of this paper that  {\it light is infinite dimensional}, or formally, that homodyne measurements over bipartite gaussian states of light cannot be understood as collective spin measurements over ensembles of independent pairs of correlated microscopic particles.

The paper is structured as follows: first, we will describe in detail a simple scenario where microscopic effects over the macroscopic variables can already be seen. We will give a first illustrative application of that to Brownian motion. Then, in Section \ref{fixed_output}, we will study the general problem of characterizing gaussian distributions generated by microscopic variables with a fixed output structure. We will treat both the finite and infinite dimensional case. We will also address some problems which may arise in practical implementations of macroscopic experiments. In Section \ref{free_output} we will analyze the case where the set of possible values of the microscopic variables is not known a priori and may actually vary between independent tuples. The complexity of the problem will be treated next, and the consequences for Foundations of Quantum Mechanics will be covered in Section \ref{quantum_non_loc}. In Section \ref{conclusion} we will present our conclusions.

\section{A simple example}
\label{simple_example}
Suppose that two parties, say Alice and Bob, perform an experiment like the one described in Fig. \ref{simple_ex}, that allows them to measure the continuous variables $X$ and $Y$, respectively. After many repetitions of the experiment, they find that $X$ and $Y$ follow a bivariate gaussian probability distribution $G(X,Y)$, characterized by the values 

\be
\langle X\rangle, \langle Y\rangle,\gamma_{XX}\equiv\langle X^2\rangle-\langle X\rangle^2, \gamma_{YY}\equiv\langle Y^2\rangle-\langle Y\rangle^2, \gamma_{XY}\equiv\langle XY\rangle-\langle X\rangle\langle Y\rangle. 
\ee

\noindent Alice and Bob then make the hypothesis that $X$ and $Y$ are not gaussian by accident, but because they actually correspond to the sum of many pairwise-correlated independent variables. A plausible explanation is that whenever they initialize their experiment, multiple particle pairs are created in some intermediate region. Suppose for simplicity that the particles of each pair are identical and have only two levels (Figure \ref{simple_ex}). That is, whenever some particle impinges on a detector, it will release a signal of strength $\lambda_2$ ($\lambda_1$) if the particle happened to be on state $2$ ($1$)\footnote{This is the case, for instance, when the particles happen to be photons and the detector is composed of a vertically disposed polarizer followed by a photocounter. In that case, photons with vertical polarization (state $2$) will reach the photocounter and thus produce $\lambda_2=1$ electrons, whereas no electrons ($\lambda_1=0$) will be released if the photon was horizontally polarized (state $1$).}; $X$ and $Y$ are just the sum of all such individual signals\footnote{We assume the detectors' responses to be linear.}. We will call this theory the \emph{microscopic dichotomic model}. Alice and Bob could propose this model, for instance, to explain why the outcomes of homodyne measurements in usual quantum optical experiments follow gaussian distributions. 

We will next show that, in certain circumstances, the microscopic dichotomic model can be experimentally refuted, even when no assumptions are made on the values of $\lambda_1,\lambda_2$, the number of particles $N$ and even allowing for the existence of different types of particles, each of them dichotomic, associated to different values of $\lambda_i$.

\begin{figure}
  \centering
  \includegraphics[width=13 cm]{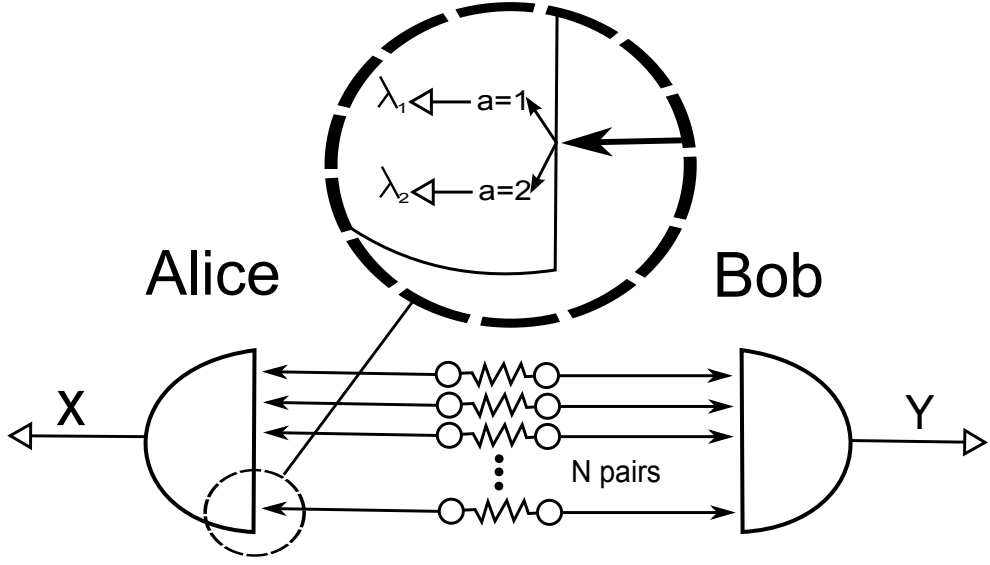}
  \caption{\textbf{Microscopic model of a macroscopic experiment.} N particle pairs travel to Alice's and Bob's lab and interact linearly with their detectors. The macroscopic variables $X,Y$ observed by Alice and Bob are the coarse-grained sum of all such microscopic interactions.}
  \label{simple_ex}
\end{figure}

We suppose then that our probability distribution follows from a microscopic dichotomic  model. In this model we allow for the existence of several types of independent pairs of correlated particles and for each of these types of particles we assume a dichotomic model. 

The contribution to the expectation values and covariance matrix of each pair of particles $j$ is given by 

\begin{eqnarray*}
&&\langle X\rangle^j= \sum_ a p_A^j(a)\lambda^j_{a} \\
&&\langle Y\rangle^j= \sum_b p_B^j(b)\lambda^j_{b}
\\
&&\gamma_{XX}^j= \sum_{a=1,2}\sum_{a'=1,2} \lambda^j_{a}\lambda^j_{a'}\{p_A^j(a)\delta(a,a')-p_A^j(a)p_A^j(a')\}\\
&&\gamma_{YY}^j=\sum_{b=1,2}\sum_{b'=1,2} \lambda_{b}^j\lambda_{b'}^j\{p_B^j(b)\delta(b,b')-p_B^j(b)p_B^j(b')\}\\
&&\gamma_{XY}^j=\sum_{a,b=1,2} \lambda_{a}^j\lambda_{b}^j\{p^j_{AB}(a,b)-p_{A}^j(a)p_B^j(b)\},
\end{eqnarray*}

\noindent where $p^j_{AB}(a,b)$ denotes the probability that in one of the pairs of type $j$, Alice's particle is in state $a$ and Bob's, in state $b$; $p^j_A(a),p^j_B(b)$ represent the corresponding marginal probabilities and  $\lambda_{a}^j, \lambda_{b}^j$ are the corresponding possible values. 

\smallskip

\noindent Now, we assume that every pair of particles is independent from the others. Therefore, calling  $N_j$ to  the number of pairs of particles of type  $j$, we have  that the expectations and the values of the covariance matrix $\gamma$ are given by

\begin{eqnarray}
&&\langle X\rangle= \sum_{j} N_j \sum_ a p_A^j(a)\lambda^j_{a}\nonumber\\
&&\langle Y\rangle= \sum_{j}N_j \sum_b p_B^j(b)\lambda^j_{b}\nonumber\\
&&\gamma_{XX}=\sum_{j} N_j\sum_{a=1,2}\sum_{a'=1,2} \lambda^j_{a}\lambda^j_{a'}\{p_A^j(a)\delta(a,a')-p_A^j(a)p_A^j(a')\},\nonumber\\
&&\gamma_{YY}=\sum_{j} N_j\sum_{b=1,2}\sum_{b'=1,2} \lambda_{b}^j\lambda_{b'}^j\{p_B^j(b)\delta(b,b')-p_B^j(b)p_B^j(b')\},\nonumber\\
&&\gamma_{XY}=\sum_{j} N_j\sum_{a,b=1,2} \lambda_{a}^j\lambda_{b}^j\{p^j_{AB}(a,b)-p_{A}^j(a)p_B^j(b)\}.
\label{matriz}
\end{eqnarray}

\smallskip

The values of $\langle X\rangle,\langle Y\rangle$ do not carry much information about the original microscopic distribution, since we can set them to any value while leaving $\gamma$ invariant. Indeed, suppose that we add to the ensemble a species of particles such that $\lambda_1=0,\lambda_2\not=0$. Then, setting $p(1,2)=1$ and varying the value of $\lambda_2$ (or, equivalently, the number of particles of such species present in the ensemble), we can vary the value of $\langle X\rangle$ without affecting $\langle Y\rangle$ or $\gamma$. Likewise, we can choose the value of $\langle Y\rangle$.

We will search for the information about the origin of $X$ and $Y$ hidden in $\gamma$. Concerning this, one can check that the following proposition holds:

\begin{prop}
Let $G(X,Y)$ be a gaussian bivariate probability distribution.  Then, $G(X,Y)$ can be generated by a microscopic dichotomic model (as defined above) iff

\begin{eqnarray}
&&\gamma_{XX}-|\gamma_{XY}|\geq 0,\nonumber\\
&&\gamma_{YY}-|\gamma_{XY}|\geq 0.
\label{basic}
\end{eqnarray}

\end{prop}

\begin{proof}
Let us first prove the necessity of eqs. (\ref{basic}). By convexity, it is enough to see that the inequalities hold for each of the terms $j$ of the sums in eq. (\ref{matriz}). Then, written in terms of $p_A(1),p_B(1),p_{AB}(1,1)$, the contribution of each term to $\gamma_{XX}$ and $\gamma_{XY}$ (omitting $j$) is equal to $(\lambda_1-\lambda_2)^2p_A(1)[1-p_A(1)]$ and $(\lambda_1-\lambda_2)^2[p_{AB}(1,1)-p_A(1)p_B(1)]$. Subtracting both expressions -and ignoring the non-negative $(\lambda_1-\lambda_2)^2$ factor-, we end up with 

\begin{eqnarray}
&p_A(1)[1-p_A(1)]-p_{AB}(1,1)+p_A(1)p_B(1)\geq\nonumber\\
&\geq p_{AB}(1,1)[1-p_A(1)]-p_{AB}(1,1)+p_A(1)p_B(1)=\nonumber\\
&=p_A(1)[p_B(1)-p_{AB}(1,1)]\geq 0.
\end{eqnarray}

We have thus demonstrated that $\gamma_{XX}-\gamma_{XY}\geq 0$. The rest of the inequalities are proven analogously.

To show the opposite implication, suppose that $\gamma$ is such that the inequalities (\ref{basic}) are satisfied, and assume that $\gamma_{XY}\geq 0$ (the case $\gamma_{XY}\leq 0$ is very similar). Then, 

\be
\gamma=\gamma_{XY}\left(\begin{array}{cc}1&1\\1&1\end{array}\right) + \left(\begin{array}{cc}\gamma_{XX}-\gamma_{XY}&0\\0&\gamma_{YY}-\gamma_{XY}\end{array}\right).
\ee

The first covariance matrix on the right-hand side of the above equation can be attained by a species of particles following the microscopic distribution $p_A(c)=p_B(c)=p_{AB}(c,c)=1/2$, for $c=0,1$. The diagonal matrix corresponds to a situation where the particles composing the pair are independently distributed. Adding to the previous ensemble a new species of particle pairs with 

\begin{eqnarray}
&(\lambda'_1-\lambda'_2)^2p'_A(0)(1-p'_A(0))=\gamma_{XX}-\gamma_{XY},\nonumber\\
&p'_B(0)(1-p'_B(0))(\lambda'_1-\lambda'_2)^2=\gamma_{YY}-\gamma_{XY},\nonumber\\
&p'_{AB}(0,0)=p'_{A}(0)p'_{B}(0),\nonumber
\end{eqnarray}

\noindent we thus manage to reproduce $\gamma$.

\end{proof}

Note that, for any $r,\epsilon>0$, any gaussian bivariate distribution with covariance matrix of the form

\be
\gamma_{XX}=1/r,\gamma_{XY}=1,\gamma_{YY}=r+\epsilon
\label{useful}
\ee
 
\noindent violates the inequalities (\ref{basic}) as long as $r>1$. It follows that not all bivariate gaussian distributions can be generated by sums of independent pairs of identical two-valued particles. The analysis of macroscopic continuous data can thus give us important clues about the microscopic structure that lies underneath.

On a foundational side, this result implies that Alice and Bob can design a quantum optics experiment in order to disprove that entangled gaussian states of light follow a microscopic dichotomic model. We will come back to this topic in Section \ref{uniform}.

\section{Fixed output structure}
\label{fixed_output}

Our aim here is to generalize the problem posed in the previous section. That is, given a set of gaussian variables $(X_1,X_2,...,X_n)$, we consider the problem of extracting information about the underlying microscopic model that gave rise to them. Equivalently, we are interested in how different restrictions on such a microscopic model translate to the macroscopic scale. Suppose, for instance, that we wonder if $(X_1,...,X_n)$ could derive from several independent $n$-tuples of $d$-level particles impinging on different detectors. We conjecture that each detector $k\in\{1,...,n\}$ is \emph{simple}, i.e., it produces a signal of magnitude $\lambda^k_a$ when a particle on state $a$ impinges on it (we will see later that this assumption can be relaxed). We can thus regard each particle impinging on detector $k$ as a $d$-valued microscopic variable $x_k\in\{\lambda^k_a\}_{a=1}^d$.

Suppose now that we have partial knowledge of our measurement devices. Specifically, imagine that we know the sets of values $\{\lambda^k_a\}$ up to a proportionality constant. Such is the case, for instance, when the response function of each detector is proportional to the spin of the incident particle, but we ignore the precise value of the coupling constant (see section \ref{uniform}). In that case, we could take $\lambda^k_a\propto(a-d/2)$.

For fixed $\Lambda\equiv\{\lambda^k_a:k=1,...,n;a=1,...,d\}$, we want to determine which $n$-variate gaussian distributions can be generated by independent tuples of correlated microscopic variables with spectrum proportional to $\Lambda$. In the generic case (whenever each set $\{\lambda^k_a\}_a$ has positive and negative values), it is easy to see that the  expectations $v_i=\langle X_i\rangle$ can be varied at will without modifying the $n\times n$ covariance matrix $\gamma_{ij}\equiv \langle X_iX_j\rangle-\langle X_i\rangle\langle X_j\rangle$ of the macroscopic system. As before we will focus in the information about $\Lambda$ contained in $\gamma$.

As in the previous section, we want to allow for the possible existence of different types of particles, each of them with their own values $\lambda^k_a$. Our aim, thus, is to characterize those covariance matrices that can be expressed as

\be
\gamma_{ij}\propto \sum_{l=1}^N \sum_{a,b=1}^d \lambda^{i}_{a}\lambda^{j}_{b}\{p^l_{ij}(a,b)-p^l_{i}(a)p^l_j(b)\}.
\ee

In this expression, $l$ indexes the different $n$-tuples $(x^l_1,...,x^l_n)$, and $\lambda_{a}^{i}$ denotes the $a^{th}$ possible value of the variable $x^l_i$. $p_{i,j}^l(a,b)$ is the probability that variables $x^l_i$, $x^l_j$ attain the values $\lambda^i_a$, $\lambda^j_b$, respectively. With a slight abuse of notation, by $p_{jj}(a,a')$ we mean $\delta_{aa'}p_j(a)$. Note that, for any $S,T\in \R$, the families of covariance matrices attainable with the sets of outcomes $\{\lambda^k_i\}$ and $\{S\cdot \lambda^k_i+T\}$ are the same.

In order to solve the above problem, as well as some others that will appear along the article, we have to introduce the concept of \emph{naked covariance matrix}.

\subsection{The naked covariance matrix}
\label{nakedness}

Imagine an idealized detector that, when impinged by a $d$-level particle, sends a signal of magnitude 1 to one counter or another depending on the state of such a particle. If we were to send a beam of particles of the same species to this detector, the macroscopic current $X^a_k$ measured on each counter $(k,a)$ would thus indicate the number of particles of the beam $k$ in a particular state $a$. We will call such a device a \emph{naked detector}.

Note that, if the values of $\{\lambda^k_1,..., \lambda^k_d\}$ are known, we can simulate the behavior of a simple detector by means of a naked detector, see Figure \ref{naked_dressed}. Indeed, the single current of the former on each round of experiments could be determined by summing up the currents registered on each arm of the naked detector, each with its appropriate weight. That is, $X_k=\sum_a\lambda^k_aX^a_k$. If we were able to characterize the cone of all possible covariance matrices $\Gamma^{ab}_{ij}\equiv \langle X^a_iX^b_j\rangle-\langle X^a_i\rangle\langle X^b_j\rangle$ generated by $n$-tuples of $d$-valued variables impinging on $n$ different naked detectors, we could thus describe the set of covariance matrices generated by simple detectors.

\begin{figure}
  \centering
  \includegraphics[width=13 cm]{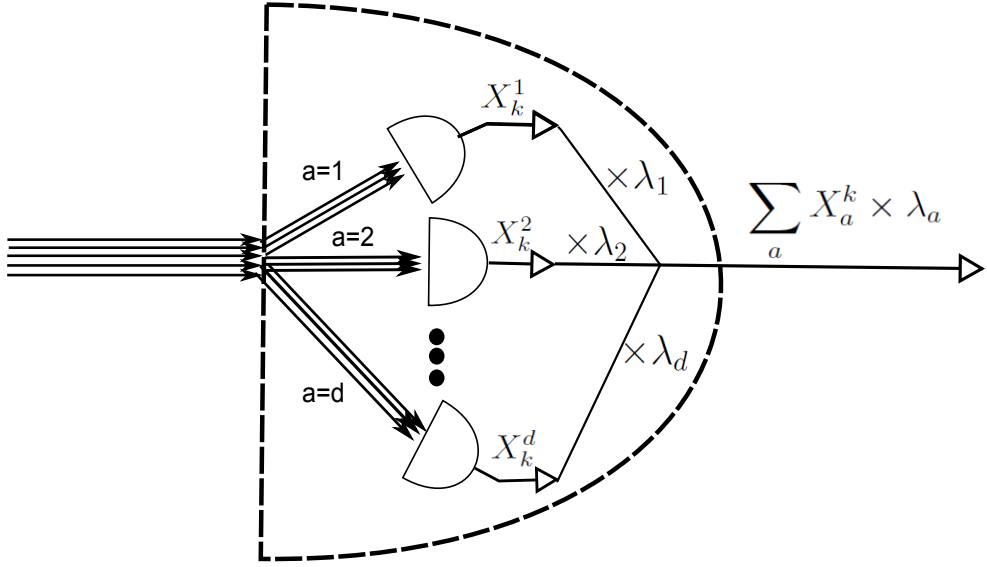}
  \caption{\textbf{Behavior of a naked detector.} A naked detector divides the original beam $k$ into sub-beams of particles in the same state $a$, and then measures the intensity $X^k_a$ of each sub-beam. As indicated in the figure, we can simulate a simple detector by postprocessing the outputs of a naked detector.}
  \label{naked_dressed}
\end{figure}

Fortunately, such a characterization is possible. The next technical result will be used extensively along the rest of the paper.

\begin{theo}
\label{naked}
Let $\Gamma$ be an $nd\times nd$ matrix. For any pair of lists of outcomes $\vec{c},\vec{c}'\in\{1,...,d\}^n$, consider the microscopic distributions

\be
p(a_1,a_2,...,a_n)=\frac{1}{2}\prod_{k=1}^n\delta(\vec{c},\vec{a})+\frac{1}{2}\prod_{k=1}^n\delta(\vec{c}',\vec{a}),
\ee

\noindent where $\delta(s,s')$ denotes the Kronecker delta, and call $\Gamma(\vec{c},\vec{c}')$ their associated naked covariance matrices.

Then $\Gamma$ is a naked covariance matrix iff it belongs to the cone generated by $\{\Gamma(\vec{c},\vec{c}'):\vec{c},\vec{c}'\in\{1,...,d\}^n$.

\end{theo}

\begin{proof}
The left implication is evident. For the opposite one, it is enough to prove that we can recover all matrices generated by a \emph{single} microscopic distribution $p_{1,...,n}(a_1,a_2,...,a_n)$. In that case, all the entries of the naked covariance matrix are of the form 

\be
\Gamma_{ij}^{a,b}\equiv\langle X_j^aX_k^b\rangle-\langle X_j^a\rangle\langle X_k^b\rangle=p_{jk}(a,b)-p_{j}(a)p_k(b), 
\ee

\noindent where, with a slight abuse of notation, by $p_{jj}(a,a')$ we mean $\delta_{aa'}p_j(a)$. Note that we can write $\Gamma_{ij}^{a,b}$ as the difference of two probabilities. Indeed, imagine that we extend our system of $n$ $d$-valued variables $(a_1,a_2,...,a_n)$ to a system of $2n$ $d$-valued variables $(a_1,a_2,...,a_n,a_{n+1},...,a_{2n})$ subject to the distribution $p(a_1,...,a_n,a_{n+1},...,a_{2n})=p_{1,...,n}(a_1,...,a_n)\cdot p_{1,...,n}(a_{n+1},...,a_{2n})$. Then, $\Gamma_{ij}^{a,b}=p_{i,j}(a,b)-p_{i,j+n}(a,b)$, and the non-linearity on the probability distribution is eliminated. Now $p(a_1,...,a_n,a_{n+1},...,a_{2n})$ is symmetric with respect to the interchange $a_i\leftrightarrow a_{i+n}$. Consequently, it can be written as a convex sum of distributions of the form

\begin{eqnarray}
q(a_1,...,a_n,a'_1,...,a'_{n})=\frac{1}{2}\delta(\vec{c},\vec{a})\delta(\vec{c}',\vec{a}')+\frac{1}{2}\delta(\vec{c}, \vec{a}')\delta(\vec{c}',\vec{a}),
\end{eqnarray}

\noindent where $\vec{c},\vec{c}'\in\{0,1,...,d-1\}^n$. The contribution of each of them to $\Gamma$ is thus proportional to

\begin{eqnarray}
q_{i,j}(a,b)-q_{i,j+n}(a,b)=&&\frac{1}{2}\{\delta(c_i,a)\delta(c_j,b)+\delta(c'_i,a)\delta(c'_j,b)-\nonumber\\
&&-\delta(c_i,a)\delta(c'_j,b)-\delta(c'_i,a)\delta(c_j,b)\}.
\end{eqnarray}

On the other hand,

\begin{eqnarray}
&\Gamma_{ij}^{ab}(\vec{c},\vec{c}')=\frac{1}{2}\{\delta(c_i,a)\delta(c_j,b)+\delta(c'_i,a)\delta(c'_j,b)\}-\nonumber\\& -\frac{1}{4}(\delta(c_i,a)+\delta(c'_i,a))(\delta(c_j,b)+\delta(c'_j,b))=\frac{1}{2}\{q_{i,j}(a,b)-q_{i,j+n}(a,b)\}.
\end{eqnarray}

\noindent $\Gamma$ therefore lies inside the cone of $\{\Gamma(\vec{c},\vec{c}')\}$.

\end{proof}

Now, let $\gamma$ be a covariance matrix generated by microscopic variables variables with known spectrum $\{\lambda^k_i\}$. Since $\gamma_{ij}=\sum_{a,b=1}^d\lambda^i_a\lambda^j_b\Gamma^{ab}_{ij}$ for some naked covariance matrix $\Gamma$, Theorem \ref{naked} tells us that $\gamma$ is inside the cone generated by the covariance matrices 

\be
\gamma_{ij}(\vec{c},\vec{c}')\equiv \sum_{a,b=1}^d\lambda^i_a\lambda^j_b\Gamma_{ij}^{ab}(\vec{c},\vec{c}')=\vec{w}(\vec{c},\vec{c}')\vec{w}(\vec{c},\vec{c}')^T,
\ee

\noindent where $\vec{w}(\vec{c},\vec{c}')\in \R^n$ satisfies $w(\vec{c},\vec{c}')_i=\lambda^i_{c_i}-\lambda^i_{c'_i}$.

The problem of determining if $\gamma$ admits a microscopic model with known outcome structure $\{\lambda^k_i\}$ can therefore be cast as a linear program.

An immediate consequence of this observation is that, as long as $d<\infty$ and $\{\lambda^k_i\}$ is fixed, there are always bivariate gaussian distributions impossible to attain with such a microscopic model. Indeed, suppose that we normalize our two macroscopic variables $X_1,X_2$ in such a way that their associated $2\times 2$ covariance matrix $\gamma$ satisfies the condition 

\be
\tr(\gamma)=1.
\label{normal}
\ee

Then, one can consider which values of $\gamma_{12}$ are attainable for a fixed value $\gamma_{11}$. For a completely general gaussian distribution, the only condition over $\gamma$ is non-negativity. This, together with the normalization condition, implies that $|\gamma_{12}|\leq\sqrt{\gamma_{11}(1-\gamma_{11})}$: in this normalized scenario, the set of attainable points $(\gamma_{11},\gamma_{12})$ is thus a circle, see Fig. \ref{eq_spaced}.

However, suppose that $X_1,X_2$ are generated by independent pairs with outcomes $\{\lambda^1_a:a=1,...,d\}, \{\lambda^2_a:a=1,...,d\}$. According to what we have seen, there exist $\mu_{\vec{c},\vec{c}'}\geq 0$ such that

\be
\gamma=\sum_{\vec{c},\vec{c}'}\mu_{\vec{c},\vec{c}'}\gamma(\vec{c},\vec{c}')
\ee

\noindent where each $\gamma(\vec{c},\vec{c}')=\vec{w}\vec{w}^T$, with

\be
\vec{w}=\left(\begin{array}{c}\lambda^1_{c_1}-\lambda^1_{c'_1}\\ \lambda^2_{c_2}-\lambda^2_{c_2'}\end{array}\right).
\ee

It is obvious that $\tr(\gamma(\vec{c},\vec{c}'))=0$ implies $\gamma(\vec{c},\vec{c}')=0$, so these pairs of $\vec{c}, \vec{c}'$ do not contribute to $\gamma$. For the rest, define $\tilde{\gamma}(\vec{c},\vec{c}')=\gamma(\vec{c},\vec{c}')/\tr(\gamma(\vec{c},\vec{c}'))$, and note that all $\tilde{\gamma}$'s are such that $\tilde{\gamma}_{11},\tilde{\gamma}_{12}$ are in the circumference $|\gamma_{12}|=\sqrt{\gamma_{11}(1-\gamma_{11})}$. Also, we have that

\be
\gamma=\sum_{\vec{c},\vec{c}'}\tilde{\mu}_{\vec{c},\vec{c}'}\tilde{\gamma}(\vec{c},\vec{c}'),
\label{convex}
\ee

\noindent and so $\sum_{\vec{c},\vec{c}'}\tilde{\mu}_{\vec{c},\vec{c}'}=1$. It follows that $\gamma$ is a convex combination of a finite number of points in the circumference of the circle. The set of all covariance matrices arising from the microscopic model $\{\lambda^k_i\}$ is therefore a polygonal inner approximation to the circle corresponding to general gaussian distributions, and so the latter contains points separated from the former as long as $d<\infty$.

\subsection{An application: two-dimensional Brownian motion}

Due to molecular collisions, a particle floating in a two-dimensional fluid will experience random kicks that will make it move in unpredictable ways. This phenomenon is known as \emph{Brownian motion}, in honor of its discoverer, the botanist Robert Brown. The first theoretical explanation of this effect is due to Albert Einstein \cite{einstein}, who also proposed to use it to estimate the size of fluid molecules.

Mathematically, if an stochastic process $B(t)_{t\geq 0}=(X(t), Y(t))_{t\geq 0}$ is a bivariate brownian motion then there exists a vector $\mu=(\mu_X , \mu_Y)\in \mathbb R^2$ and a covariance matrix $\gamma$ such that for every $t, h\geq 0$,  the increment
$B(t+h)-B(t)$ is a bivariate normal distribution with mean $h \mu$ and covariance matrix $h \gamma$. Conversely, given any such $\mu, \gamma$, we can associate to them a bivariate brownian motion.

It is well known and widely used that the one dimensional Brownian motion can be well approximated by a one dimensional random walk where each of the Bernouilli variables takes the values $\pm \lambda$ for a suitable $\lambda$. Curiously enough, the previous section implies that the same result does not hold for a two dimensional Brownian motion, no matter how the values $\lambda_X, \lambda_Y$ of the Bernuoilli variables are fixed. Indeed, note that our setting is perfectly adapted for the study of random walks, since they arise as the sum of many independent dichotomic variables.

Suppose now that we relax the definition of Brownian motion to account for the fact that in practice we cannot measure a system between arbitrarily small time intervals $h$. Then one can allow the Bernouilli distribution to vary in time, as long as the period $T$ of such a distribution satisfies $h\gg T$; in these conditions, the macroscopic distribution $B(t)-B(t-h)$ is approximately homogeneous in time. 

Even in this more general scenario, our previous results have something to say. Let $\lambda_X=\lambda_Y$. Then we can apply the results from Section \ref{simple_example}, i.e., $\gamma$ must fulfill $\tr(\gamma W_{ij})\geq 0$, for $i,j=1,2$, where 
$$W_{ij}=\begin{pmatrix} 1+(-1)^j& (-1)^i\\
(-1)^i& 1+(-1)^{j+1}\end{pmatrix}.$$

\noindent If we further want our model to verify the stronger condition that the choice of our $X,Y$ axes can be done at will\footnote{That is, for any pair of orthonormal axes $\hat{x}',\hat{y}'$ which we use to describe the movement of the particle, there exists a two-dimensional random walk (with $\lambda_X=\lambda_Y$) compatible with our macroscopic observations.}, we arrive at the extra constraints $\tr(O\gamma O^T W_{ij})\geq 0$, for all $O\in O(2)$. A bit of algebra shows that this last (necessary and sufficient) condition is equivalent to:

\be
\mbox{det}(\gamma)\geq\frac{1}{8}\tr^2(\gamma).
\ee

\subsection{The uniform case: measuring spins}
\label{uniform}

Let us study a case of physical significance: Alice and Bob perform a quantum optics experiment in which homodyne measurements are performed on each side of a bipartite gaussian state of light, returning a pair of gaussian variables $X_1,X_2$. Having discarded the microscopic dichotomic model (see Section \ref{simple_example}), they resort to more complicated discrete microscopic models in order to explain their observations. That is how they come up with the \emph{microscopic spin model}, where $X_1,X_2$ are postulated to be the result of summing up the spins of $N$ independent $(d-1)/2$-spin particle pairs, with $N\gg 1$ (note that for $d=2$ we recover the microscopic dichotomic model). That is, there is a microscopic explanation involving only $d$ dimensional quantum systems (Hilbert spaces). We assume that Alice's and Bob's detectors behave identically, i.e., whenever a particle with spin component $m$ passes through each detector, the corresponding microscopic signal will be equal to $g\cdot m$, where the value of the constant $g$ is unknown. We wonder which gaussian states are compatible with such a microscopic model.

Notice that, in this case, $\lambda^1_a=\lambda^2_a\equiv \lambda_a$, i.e. the spectrum of the microscopic variables $x_1,x_2$ is the same. Also, we can take $\{\lambda_a\}$ to satisfy $\lambda_a-\lambda_b=a-b$. Upon imposing the normalization constraint (\ref{normal}), as long as $\gamma\not=0$, the covariance matrix will admit a decomposition of the form (\ref{convex}), with

\be
\tilde{\gamma}_{11}=\frac{(c_1-c'_1)^2}{(c_1-c'_1)^2+(c_2-c'_2)^2},\tilde{\gamma}_{12}=\frac{(c_1-c'_1)(c_2-c'_2)}{(c_1-c'_1)^2+(c_2-c'_2)^2}.
\ee

Now, we know that the point $(0,0)$ of the circle is accessible [take $c_1=c_1'$]. Our aim now is to determine the closest extreme point $\tilde{\gamma}(\vec{c},\vec{c}')$ in the circumference. That is, we want to find out

\be
\gamma_{11}^*\equiv\min_{\vec{c},\vec{c}'} \{\tilde{\gamma}(\vec{c},\vec{c}')_{11}:\tilde{\gamma}(\vec{c},\vec{c}')_{11}\not=0\}.
\ee

\noindent It is not difficult to see that $\gamma_{11}^*=\frac{1}{1+(d-1)^2}$. 

\begin{figure}
  \centering
  \includegraphics[width=20 cm]{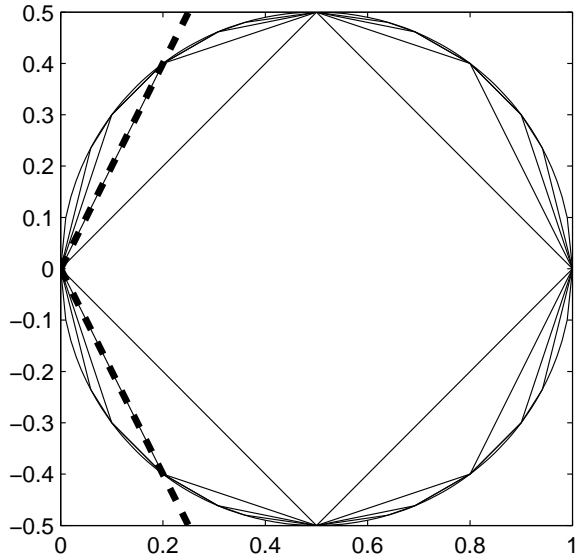}
  \caption{\textbf{Equally-spaced outcomes.} A plot of the accessible regions $(\gamma_{11},\gamma_{12})$ of normalized covariance matrices generated by microscopic variables with values $\{1,...,d\}$, for $d=2,3,4,5$. The circle with center in $(1/2,0)$ and radius $1/2$ encompassing them all represents the set of general gaussian distributions. The dashed line corresponds to inequality (\ref{spin}) for $d=3$.}
  \label{eq_spaced}
\end{figure}

Correspondingly, $\gamma_{12}^*=\pm\frac{d-1}{1+(d-1)^2}$. It follows that (see Fig. \ref{eq_spaced}) the inequality 

\be
|\gamma_{12}|\leq (d-1)\gamma_{11}
\label{spin}
\ee

\noindent holds for all gaussian distributions arising from pairs of $(d-1)/2$-spin particles, normalized or not.

On the other hand, note that this inequality is violated by all gaussian distributions with covariance matrix of the form (\ref{useful}) as long as $r>(d-1)$. 

From a foundational perspective, eq. (\ref{spin}) implies that, for any value of $d$, there exists a quantum optics experiment that proves that gaussian states of light do not follow the microscopic $(d-1)/2$-spin model. Due to the limitations of current technology, though, this claim can only be checked up to a finite value $d_t$.

\subsection{Incomplete output bases}

As $d$ grows, inequality (\ref{spin}) becomes more and more irrelevant. And, actually, one can prove that in the limit $d\to\infty$ equally spaced $d$-level systems (with positive and negative values) can reproduce any multivariate gaussian distribution. Indeed, suppose that we want to approximate the gaussian distribution $p(X_1,...,X_n)d\vec{X}$, with covariance matrix $\gamma$ and expectation vector $\vec{v}$, by means of $d$-valued microscopic variables. From previous considerations, we only have to worry about reproducing the covariance matrix $\gamma$ via microscopic ensembles. Now, $\gamma$ is positive semidefinite, and so $\gamma_{kl}=\vec{w}^k\cdot \vec{w}^l$, for some vectors $\{\vec{w}^k\}_{k=1}^n\subset \R^n$ such that $\|\vec{w}^k\|_{\infty}<L$ for some $L$. Define then the vectors $\{\vec{u}^k\}$ in such a way that $u^k_i=\frac{L}{d}\lfloor\frac{d}{L}w^k_i\rfloor$. It is clear that the new covariance matrix $\tilde{\gamma}_{kl}\equiv \vec{u}^k\cdot\vec{u}^l$ can be generated by ensembles of microscopic systems with outcomes of the form $\{\lambda_a\propto a-d:a=0,...,2d\}$. Moreover, $\lim_{d\to\infty}\tilde{\gamma}=\gamma$.

Assuming, as before, that $\lambda^i_a=\lambda^j_a$, $\forall i,j\in \{1,...,n\}$, one can use the same argument to prove that, as long as our sequence of microscopic outcomes $\{\lambda_a:a\in \N\}$ satisfies

\be
\lim_{a\to\infty}\frac{\lambda_{a+1}-\lambda_{a}}{\lambda_{a}-\lambda_{a-1}}\to 1,
\ee

\noindent any gaussian multivariate distribution can be approximated by sums of microscopic ensembles. This includes, in particular, the case $\lambda_a=f(a)$, with $f$ being a non-trivial (i.e., non-constant) rational function of $a$.

In view of this, it is natural to ask whether there exist sequences of infinitely many different outcomes $\{\lambda_a:a\in \N\}$ which do not allow to approximate certain gaussian distributions. 

At first glance, one would be tempted to answer this question in the negative: an infinity of possible outcomes gives us infinitely many degrees of freedom to play with. In such circumstances, it is difficult to see how, for a given covariance matrix $\gamma$, there could \emph{not} be a combination of infinitely many weights that reproduces or approximates $\gamma$ at the microscopic scale. On second thoughts, though, it could be that if the different outcomes are not sufficiently `spread out', then they could not generate every distribution. 

This is actually the case, as shown in Fig. \ref{expon}, when the outcomes happen to be the terms of a geometric sequence. The next proposition states it more clearly:

\begin{prop}
Let $\tilde{G}_\mu$ be the set of all normalized bivariate covariance matrices $\gamma$ arising from sums of pairs of variables with spectrum $\{\lambda_a=\mu^a:a\in \N\}$. If $\mu>\frac{3+\sqrt{2}}{2}$, then the angle $\theta$ of the extreme points with respect to the horizontal axis (see Figure \ref{expon}) satisfies

\be
\sin(\theta)\not\in \Theta_\mu\equiv\left(\frac{2(\mu-1)}{(\mu-1)^2+1},\frac{2\mu(\mu-1)}{(\mu-1)^2+\mu^2}\right).
\label{forbidden}
\ee
\end{prop}

\noindent Note that, in the limit $\mu\to\infty$, $\sin(\theta)\not\in (0,1)$, i.e., the convex body collapses to the `bivalued' square given by eqs. (\ref{basic}). Notice also that this results holds too when the outcomes are taken from the bigger set $\{\mu^a:a\in\Z\}$.

\begin{figure}
  \centering
  \includegraphics[width=16.5 cm]{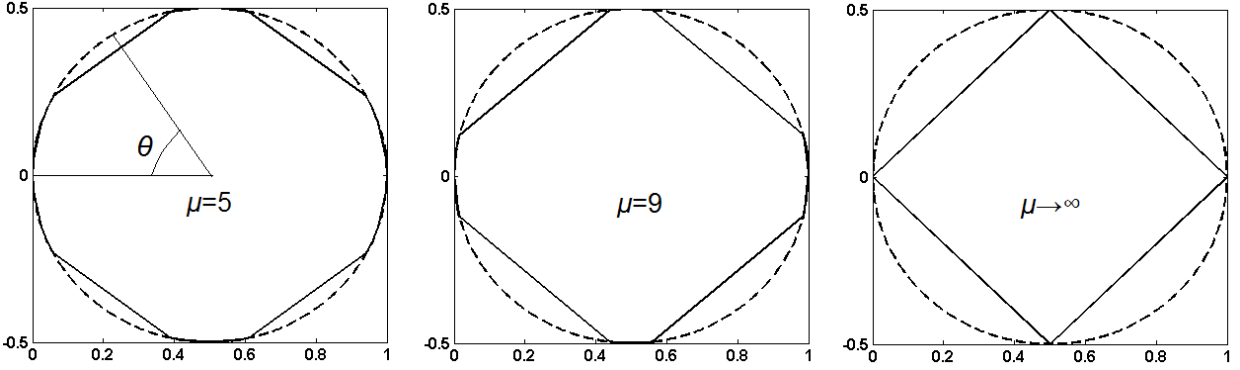}
  \caption{\textbf{Geometric outcomes.} A numerical plot of the accessible points $(\gamma_{11},\gamma_{12})$ of normalized covariance matrices generated by microscopic variables with values $\{\mu^a,a\in \Z\}$, for $\mu=5,9,\infty$. The dashed line delimits the circle of all normalized gaussian distributions. The first picture also depicts an extreme point of the circle with $\sin(\theta)\in \Theta_\mu$.}
  \label{expon}
\end{figure}

\begin{proof}
Any bivariate covariance matrix $\gamma\in\tilde{G}_\mu$ is a convex combination of normalized matrices $\gamma'$ of the form

\be
\gamma'_{11}\propto (\mu^{a}-\mu^{a'})^2,\gamma'_{22}\propto (\mu^{b}-\mu^{b'})^2,\gamma'_{12}\propto (\mu^{a}-\mu^{a'})(\mu^{b}-\mu^{b'}),
\ee

\noindent with $a,a',b,b'\in\N$ and $\gamma'_{11}+\gamma'_{22}\not=0$. We will next prove that, for any such matrix $\frac{\gamma'_{11}}{\gamma'_{22}}\not\in (1/(\mu-1)^2,(\mu-1)^2/\mu^2)$: the statement of the proposition will then follow from trivial calculations.

W.l.o.g. we can assume that $a> a',b> b'$, and then the constraint translates as $(\mu^a-\mu^{a'})/(\mu^b-\mu^{b'})\not\in(1/(\mu-1),(\mu-1)/\mu)$. There are two possibilities: either $a=b+\Delta$, with $\Delta\geq 0$, or  $a=b-\Delta'$, with $\Delta'>0$.

Let us examine the first one: if $a=b+\Delta$, for $\Delta\in \N$, then

\be
\frac{\mu^a-\mu^{a'}}{\mu^b-\mu^{b'}}=\mu^\Delta\frac{1-\frac{1}{\mu^{a-a'}}}{1-\frac{1}{\mu^{b-b'}}}\geq\left(1-\frac{1}{\mu}\right)\mu^\Delta\geq\frac{\mu-1}{\mu},
\ee

\noindent and so $\gamma_{11}'/\gamma_{22}'\not\in (1/(\mu-1),(\mu-1)/\mu)$.

Let us thus explore the other option, namely, that $a=b-\Delta'$, with $\Delta'\geq 1$. Then, we have that

\be
\frac{\mu^a-\mu^{a'}}{\mu^b-\mu^{b'}}= \mu^{-\Delta'}\frac{1-\frac{1}{\mu^{a}-\mu^{a'}}}{1-\frac{1}{\mu^{b}-\mu^{b'}}}\leq\frac{\mu}{\mu^{\Delta'}(\mu-1)}\leq \frac{1}{\mu-1},
\ee

\noindent and hence $\gamma_{11}'/\gamma_{22}'\not\in (1/(\mu-1),(\mu-1)/\mu)$.

\end{proof}

\noindent We will say that the set of values $\{\lambda_a=\mu^a:a\in \N\}$ constitutes an \emph{incomplete output basis}, since ensembles of microscopic systems with such an spectrum cannot be used to generate arbitrary covariance matrices.

In the last two subsections we have been considering symmetric scenarios where $\lambda^k_a=\lambda^{k'}_a$, for $k\not=k'$, i.e., where all microscopic systems have the same spectrum. However, the asymmetric case is also worth studying, and could play a role in future experiments to reveal or refute hidden quantization. For one, in asymmetric scenarios it is not necessary to resort to geometric sequences in order to find instances of incomplete output bases of infinite cardinality. Take, for example, the pair of spectra $\{\lambda^1_a=a:a=1,2,...\}$, $\{\lambda^2_b=1/b:b=1,2,...\}$. Proving that the condition $|\gamma_{12}|\leq \gamma_{11}$ holds in this case is left to the reader as an exercise.

\subsection{Imperfect detectors}

All previous results hold in case our detector acts in a deterministic way, that is, when our detector assigns a given signal $\lambda_a$ to each particle in state $a$. However, such is a quite unrealistic model for a typical detector behavior. One should rather expect that, depending on the state $a$ of the incident particle, our detector sends a signal of strength $\lambda$, following a certain probability distribution $\rho_a(\lambda)$. Coming back to the case where our detector consists of a polarizer connected to a photocounter, the corresponding probability distributions would be $\rho_1(\lambda)=\delta(\lambda), \rho_2(\lambda)=\eta\delta(\lambda-1)+(1-\eta)\delta(\lambda)$, where $\eta$ is the efficiency of the photocounter.

In this section we will show, though, that, at the level of covariance matrices, any probabilistic detector can be simulated by a deterministic one. Therefore, restricting our analyses to this latter case is more than justified.

Suppose, indeed, that each particle or microscopic variable $x_k$ is evaluated by an imperfect detector $k$, via the following correspondence $a\to \rho_a^k(\lambda)$. Let $p(X_1,...,X_k)$ be the gaussian distribution with covariance matrix $\gamma$ generated by the microscopic distribution $p(x_1,...,x_n)$. Defining $\langle \lambda\rangle_a^k$, $\langle \lambda^2\rangle_a^k$ as $\int d\lambda\rho_a^k(\lambda)\lambda$, $ \int d\lambda\rho_a^k(\lambda)\lambda^2$, respectively, it is straightforward that

\begin{eqnarray}
\gamma_{ij}\propto && \sum_{a,b}\left[p_{ij}(a,b)-p_i(a)p_j(b)\right]\langle\lambda\rangle_a^i\langle\lambda\rangle_b^j,\mbox{ if } i\not=j,\nonumber\\
&&\sum_{a,a'}\left[p_i(a)\delta_{aa'}-p_i(a)p_i(a')\right]\langle\lambda\rangle_a^i\langle\lambda\rangle_{a'}^i+K_i,\mbox{ otherwise},
\end{eqnarray}

\noindent where 

\be
K_i\equiv\sum_{a}p_i(a)\left(\langle\lambda^2\rangle^i_a-(\langle\lambda\rangle^i_a)^2\right)\geq 0.
\ee

It is therefore clear that if we replace each detector $k$ by the deterministic detector $a\to\langle\lambda\rangle_a^k$, we obtain a covariance matrix $\gamma'$ such that $\gamma-\gamma'$ is a diagonal positive semidefinite matrix. As long as for each $k$ there exists $a(k),b(k)\in\{1,...,d\}$ such that $\langle\lambda\rangle^k_{a(k)}\not=\langle\lambda\rangle^k_{b(k)}$ (and this condition is necessary for the macroscopic variable $X_k$ to be correlated to any other), one can then recover $\gamma$ by adding to the previous ensemble a certain number of $n$-tuples of microscopic variables independently distributed with $x_k\in\{\langle\lambda\rangle^k_{a(k)},\langle\lambda\rangle^k_{b(k)}\}$.

An immediate consequence of this correspondence between deterministic and probabilistic detectors is that relations (\ref{basic}) also hold for the latter class, provided that the two detectors involved in the experiment are identical.

\subsection{Variable number of variables}
\label{non_constant}

Till now, we have been assuming that our macroscopic variables $(X_1,...,X_n)$ are the result of adding up a given number of microscopic variables. Such an assumption, though, may not hold in reality: consider, for instance, an experimental situation where the process that triggers the production of particle $n$-tuples is the decay of a microscopic system (say, an atom). Imagine also that each atom in our sample decays during the experiment with probability $\nu<1$. In that case, our macroscopic signal would show no contributions from a particular atom with probability $1-\nu$. We could model the above situation by extending the detector response function (i.e., $\{\lambda_a\}_{a=1}^d\to \{\lambda_a\}_{a=1}^d\cup\{0\}$). That would be a gross simplification of the original setup, though, since it would allow the detectors to `not click' independently of each other. Fortunately, there is a better way to study this scenario.

Suppose that a given tuple of microscopic variables $x_1,...,x_n\in\{\lambda^1_a\}\times...\times\{\lambda^n_a\}$, with $\lambda^k_a\not=0$ for any $k,a$ is produced with probability $\nu$. Then one can check that their contribution to the macroscopic covariance matrix is of the form

\be
\gamma_{ij}\propto \nu\langle x_ix_j\rangle-\nu^2\langle x_i\rangle\langle x_j\rangle.
\ee

\noindent When $\nu=1$, we recover the usual formula, while for $\nu\ll 1$ we obtain a term proportional to $\langle x_ix_j\rangle=\sum_{a,b} p_{ij}(a,b)\lambda_a\lambda_b$. Clearly, an intermediate situation $\nu\sim 1$ is thus given by a conical combination of these two extreme cases.

If the value of $\nu$ is completely unknown and is allowed to vary between different tuples, any accessible covariance matrix $\gamma$ will hence admit a conical decomposition of the form

\be
\gamma_{ij}=\sum_{a,b,\vec{c},\vec{c}'}\mu_{\vec{c},\vec{c}'}\Gamma_{ij}^{a,b}(\vec{c},\vec{c'})\lambda^i_a\lambda^j_b+\sum_{a,b,\vec{c}}\mu_{\vec{c}}\delta(a,c_i)\delta(b,c_j)\lambda^i_a\lambda^j_b.
\label{inter_milan}
\ee

\noindent The problem of deciding if $\gamma$ can be generated by ensembles of microscopic variables which appear probabilistically in the sum defining the macroscopic variables can thus again be formulated as a linear program. Notice that the first (second) summand in (\ref{inter_milan}) shall be neglected if we postulate that $\nu\ll 1$ ($\nu\approx 1)$. Unless otherwise specified, along the rest of the paper we will always work under the assumption that $\nu=1$.

\section{Free output structure}
\label{free_output}
We have just studied the case where our macroscopic variables $X_1,...,X_n$ are generated by sums of independent $n$-tuples of microscopic variables with a known structure of values $\{\lambda^k_a:a=1,...,d, k=1,...,n\}$. Similarly, one could envision a related scenario where such microscopic variables have $d$ levels, but the concrete values associated to each level are unknown, and may even vary between the different tuples. The problem we propose is thus to characterize which gaussian distributions can arise from generic averages of $d$-level microscopic systems.

As it turns out, if we allow the set of outcomes $\{\lambda^{k,l}_a:a=1,...,d, k=1,...,n\}$ for each $n$-tuple $l$ to be completely arbitrary, then any $n$-variate gaussian distribution can be generated with 2-valued microscopic systems. Indeed, let $p(X_1,...,X_n)$ be a gaussian distribution with displacement vector $\vec{v}=\langle\vec{X}\rangle$ and covariance matrix $\gamma$. Since $\gamma\geq 0$, it admits a Gram decomposition, i.e., there exists a set of vectors $\{\vec{u}^i\}_{i=1}^n\subset \R^n$ such that $\gamma_{ij}=\vec{u}^i\cdot \vec{u}^j$. Now, define $\gamma^l_{ij}\equiv u^i_lu^j_l$, and note that 

\be
\gamma^l_{ij}\propto\sum_{a,b=1,2}\Gamma^{ab}_{ij}(\vec{1},\vec{2})\lambda^{i,l}_a\lambda^{j,l}_b,
\ee

\noindent with the values $\lambda^{k,l}_1=u^k_l, \lambda^{k,l}_2=0$. Here $\vec{1}$ [$\vec{2}$] denotes a vector of the form $(1,1,1,...)$ [$(2,2,2,...)$].

$\gamma$ is thus a conical combination of covariance matrices $\{\gamma^l\}_{l=1}^n$ of 2-valued microscopic variables. On the other hand, the displacement vector $\vec{v}$ can be modified by an arbitrary amount $\Delta\vec{v}$ without altering $\gamma$ by adding a constant $n$-tuple $(\Delta v_1,...\Delta v_n)$ to the ensemble. The initial gaussian distribution $p(X_1,...,X_n)$ can thus be completely recovered.

Nevertheless, between a complete knowledge of the outcome structure and a complete ignorance there exist natural intermediate situations where the problem of characterizing the resulting gaussian distributions becomes non-trivial. Suppose, for instance, that we impose the additional hypothesis that the microscopic variables in each $n$-tuple have the same structure, i.e., we postulate that $\lambda^{k,l}_a=\lambda^{k',l}_a$ for all $k,k'\in\{1,...,n\}$ and all $l$. This is a natural assumption when the sequence $X_1,X_2,...,X_n$ represents measurements of the same macroscopic variable $X$ at different times, i.e., $X_k=\sum_{l=1}^Nx^{(l)}(t_k)$. In such a situation, it is reasonable to postulate that the probability distribution of each microscopic variable $x^{(l)}$ evolves with time, while its set of possible values remains constant. This scenario may appeal to those interested in computational biology, where simple (i.e., not memory consuming) discrete mathematical idealizations of biological entities are sought to fit macroscopic time series. Also, a closely related assumption will be used in Section \ref{quantum_non_loc} to rule out the existence of finitely-valued classical models for certain families of quantum experiments.

From a mathematical point of view, the hypothesis of identical outcomes makes the problem non-trivial again: indeed, using Gram decomposition arguments, it is easy to infer that any multivariate gaussian distribution $p(X_1,...,X_n)$ can be generated by summing independent $n$-tuples of $n+1$-valued identical microscopic systems. The next proposition shows that the upper bound $n+1$ on the minimal number of outcomes is actually optimal.

\begin{prop}
\label{identity}
Let $p(X_1,...,X_n)$ be an arbitrary gaussian distribution with covariance matrix $\gamma$ generated by convolutions of microscopic distributions $p(x^{(l)}_1,...,x^{(l)}_n)$, with $x^{(l)}_k\in\{\lambda^{(l)}_a\}_{a=1}^n,\forall k$. Then,

\be
B(\gamma)\equiv\left(1-\frac{3}{n(4^n-1)}\right)\sum_{i=1}^n\gamma_{ii}-\frac{3}{4^n-1}\sum_{i,j=1}^n2^{i+j-2}\gamma_{ij}\geq 0.
\label{restrict}
\ee

\end{prop}

\noindent Note that the former inequality is violated by general gaussian distributions (take, for instance, $\gamma_{ij}=2^{i+j}$).

To prove the proposition, we will need the following result. 

\begin{lemma}
\label{technical}
Let $D=\{\vec{d}^k\}_{k=1}^m\subset \R^n$ be a set of vectors of the form $d_i^k=\delta(i,c_k)-\delta(i,c'_k)$, for $\vec{c},\vec{c}'\in\{1,2,...,n\}^m$. If the vectors in $D$ are linearly dependent, then there exists a non-null vector $\vec{x}\in\{-1,0,1\}^m$ such that $\sum_{k=1}^mx_k\vec{d}^k=0$.
\end{lemma}

\begin{proof}
First, some notation: given a non-null vector $\vec{d}\in \R^n$ of the form $d_i=\delta(i,c)-\delta(i,c')$, we will call $c$ and $c'$ the \emph{occupied indices of $\vec{d}$}. Also, we will denote by $\{\ket{i}:i=1,...,n\}$ the canonical basis of $\R^n$, i.e., $\ket{1}=(1,0,0,...)$, $\ket{2}=(0,1,0,...)$, etc.

By hypothesis, there exist coefficients $\{\mu_k\}$ such that $\sum_{k}\mu_k\vec{d}^{k}=0$. Consider the set of vectors $E=\{\vec{d}^k:\mu_k\not=0\}$. If there exists $\vec{e}\in E$ such that $\vec{e}=0$, then we have that $1\cdot\vec{e}=0$, and we have finished. Suppose, on the contrary, that none of the elements of $E$ is null. Define $x_1=1$ and choose an arbitrary element of $E$ (denote it $\vec{e}^1$), with $\vec{e}^1=\ket{i_1}-\ket{i_2}$. Since $\vec{e}^1\not=0$, there must exist another non-null vector $\vec{e}^2$ sharing an occupied index with $\vec{e}^1$ (for otherwise $\sum_{k}\mu_k\vec{d}^{k}=0$ would not hold). Call  $(i_2,i_3)$ the occupied indices of $\vec{e}^2$. Take $x_2=e^2_{i_2}$. If $i_1=i_3$, then $\sum_{i=1,2}x_i\vec{e}^i=\ket{i_1}-\ket{i_2}+\ket{i_2}-\ket{i_1}=0$, and we have finished. If such is not the case, then there is a third vector $\vec{e}^3$ whose occupied indices are $(i_3,i_4)$. Again set $x_3=e^3_{i_3}$. If $i_4=i_1$ or $i_4=i_2$, then either $\sum_{i=1}^3x_i\vec{e}^i=0$ or $\sum_{i=2}^3x_i\vec{e}^i=0$, respectively. If $i_4$ does not equal any of the previous indices, then there exists a different vector $\vec{e}^4\in E$ with occupied indices $(i_4,i_5)$, etc. Since the set of rows is finite, if we iterate the procedure at some point we will find a vector whose occupied indices $(i_k,i_{k+1})$ are such that $i_{k+1}=i_{k'}$, for some $k'<k$, in which case we know that $\sum_{j=k'}^kx_j\vec{e}^j=\sum_{j=k'}^k\ket{i_j}-\ket{i_{j+1}}=0$. 
\end{proof}

\noindent \emph{Proof of Proposition \ref{identity}.} In order to prove relation (\ref{restrict}), it is enough to show that it holds for covariance matrices of the form $\gamma=\vec{w}\vec{w}^T$, with the vector $\vec{w}\in \R^n$ given by

\be
w_i\equiv(\lambda_{c_i}-\lambda_{c'_i}),
\ee

\noindent for some $\vec{c},\vec{c}'\in\{1,...,n\}^n$. Our next step will be to prove that, for any choice of $\vec{c},\vec{c}'$, there exist a pair of disjoint sets $E,F\in\{1,2,...,n\}$ (one of which may be actually empty) such that the non-null vector $s_i=\chi_{E}(i)-\chi_{F}(i)$ satisfies $\vec{s}\cdot\vec{w}=0$.

Indeed, note that $\vec{w}=C\vec{\lambda}$, where $C_{ik}=\delta(c_i,k)-\delta(c'_i,k)$.  Now, the rows of $C$ must be linearly dependent, since there are $n$ of them and all are perpendicular to the vector $(1,1,...,1)$. Applying Lemma \ref{technical} to $C$'s rows we thus have that there exist a non-null vector $\vec{s}\in \{-1,0,1\}^n$ such that $\vec{s}^TC=0$.

Now, $B(\vec{w}\vec{w}^T)$ can be expressed as

\be
B(\vec{w}\vec{w}^T)=\left(1-\frac{1}{n\|\vec{v}\|^2}\right)\|\vec{w}\|^2-\frac{1}{\|\vec{v}\|^2}(\vec{w}\cdot\vec{v})^2,
\label{interm}
\ee

\noindent where $\vec{v}$ corresponds to the vector $v_i=2^{i-1}$. Let $\vec{s}\in\{-1,0,1\}^n$ be the non-null vector such that $\vec{s}\cdot\vec{w}=0$, and call $E$ and $F$ the sets of indices where its entries are $1$ or $-1$, respectively. Then we have that $|\vec{v}\cdot\vec{s}|=|\sum_{i\in E}2^i-\sum_{i\in F}2^i|\geq 1$, since $\sum_{i\in E}2^i$ and $\sum_{i\in F}2^i$ correspond to the binary expansion of two different natural numbers. Call $\hat{s}$ the normalization of the vector $\vec{s}$. Then, $\vec{v}=\mu\hat{s}+\sqrt{\|\vec{v}\|^2-\mu^2}\hat{s}^\perp$, where $\hat{s}^\perp$ is a unit vector orthogonal to $\hat{s}$ and $|\mu|=|\vec{v}\cdot\hat{s}|\geq \frac{1}{\sqrt{|E|+|F|}}\geq\frac{1}{\sqrt{n}}$.

Finally, we have that 

\be
|\vec{w}\cdot\vec{v}|=\sqrt{\|\vec{v}\|^2-\mu^2}|\hat{s}^\perp\cdot\vec{w}|\leq \sqrt{\|\vec{v}\|^2-\frac{1}{n}}|\hat{s}^\perp\cdot\vec{w}|\leq \sqrt{\|\vec{v}\|^2-\frac{1}{n}}\|\vec{w}\|.
\ee

\noindent Substituting in (\ref{interm}), we arrive at (\ref{restrict}).\begin{flushright}$\square$\end{flushright}

\subsection{How different are general and finitely generated covariance matrices?}

Call ${\cal G}^d_n$ the cone of all covariance matrices generated by $n$-tuples of identical $d$-valued microscopic variables. Given that both ${\cal G}^d_n$ and the set ${\cal G}^{\infty}_n$ of $n$-variate general covariance matrices are cones, estimating their difference does not make sense unless a scale is fixed. A natural way to do so is via normalization, i.e., by dilating each covariance matrix $\gamma$ until it satisfies $\tr(\gamma)=1$. Intuitively, this normalization constraint fixes the total amount of noise in the system to be equal to unity. 

The next step is to define a suitable distance to quantify the difference of the sets $\G^d_n,\G^\infty_n$. One possibility is to use a witness $G\in M_{n\times n}$. That way, for fixed $G$, a measure of the difference between $\G^d_n,\G^\infty_n$ would be given by

\be
K^d(G)\equiv \frac{\max\{|\tr(G\gamma')|:\gamma'\in \G^d_n,\tr(\gamma')=1\}}{\max\{|\tr(G\gamma)|:\gamma\in \G^\infty_n, \tr(\gamma)=1\}}.
\ee

\noindent Clearly, $K^d(G)\leq 1$, and the lower its value, the greater the distance between the two sets.

One way to interpret eq. (\ref{restrict}) is that there exists a witness $G\in M_{n\times n}$ such that

\be
K^n(G)=1-\frac{3}{n(4^n-1)}.
\label{parecido}
\ee

\noindent Indeed, take $G_{ij}=2^{i+j-2}\frac{3}{4^n-1}$. $G$ itself is a rank-1 normalized positive semidefinite matrix, and so a normalized element of $\G^\infty_n$. It follows that, for all normalized $\gamma$, $\tr(G\gamma)\leq 1$, with equality for $\gamma=G$. On the other hand, according to eq. (\ref{restrict}), $\tr(G\gamma')\leq 1-\frac{3}{n(4^n-1)}$, for all normalized $\gamma'\in \G^n_n$ (and it is not difficult to see that this inequality can be saturated). 

Relation (\ref{parecido}) suggests that, although different, ${\cal G}^d_n$ and ${\cal G}^{\infty}_n$ are exponentially close, and thus they would be very difficult to distinguish in practice. However, since general $n$-variate gaussians can be generated by $n+1$-dimensional systems, one could argue that the proximity between ${\cal G}^d_n$, ${\cal G}^{\infty}_n$ is just a restatement of the fact that $n$ can be approximated by $n+1$ in the limit of large $n$. And actually, if free outcome structure models are to be of any practical use, the relevant question is instead how close ${\cal G}^d_n$ and ${\cal G}^{\infty}_n$ are in the limit $n\gg d\gg 1$.

To shed light on this matter, we will follow the lines of \cite{jop} and consider an idealized scenario where the macroscopic variables form a continuum, i.e., our gaussian variables are $\{X_t: t\in [0,1]\}$. In these conditions, the covariance matrix of our system shall be replaced by a positive semidefinite kernel of the form $\gamma(t,u)$, such that 

\be
\int_{{\cal R}} (\langle X_t X_u\rangle-\langle X_t\rangle \langle X_u\rangle) dtdu=\int_{{\cal R}}\gamma(t,u)dtdu,
\ee

\noindent where ${\cal R}$ is an arbitrary region of $[0,1]^2$. We will further assume that our macroscopic variables are normalized so that

\be
\int_{0}^1\gamma(t,t)dt=1,
\ee

\noindent this condition being the continuum counterpart of demanding that the trace of the covariance matrix is equal to 1.

In this setting, any normalized positive-semidefinite symmetric kernel qualifies as the covariance kernel of general macroscopic variables. On the other hand, the kernels generated by microscopic variables with spectrum of cardinality $d$ belong to the set

\be
\G^d\equiv\overline{\mbox{cone}}\left\{\gamma(t,u):\gamma(t,u)=\sum_{k,j=1}^{d}\lambda_k\lambda_j(\chi^A_k(t)-\chi^B_k(t))(\chi^A_j(u)-\chi^B_j(u))\right\},
\label{contin_d}
\ee 

\noindent where $\chi^A_k$ ($\chi^B_k$) denotes the characteristic function of the set $A_k\subset [0,1]$ ($B_k\subset [0,1]$). The sets $\{A_k\}_{k=1}^{d}$ ($\{B_k\}_{k=1}^{d}$) are mutually disjoint, and satisfy $\cup_{k=1}^{d} A_k=[0,1]$ ($\cup_{k=1}^{d} B_k=[0,1]$).

In this continuum scenario, witnesses $G\in M_{n\times n}$ shall be replaced by bounded symmetric kernels $G(u,t)$, and the action of a witness $G$ over an element $\gamma$ of $\G^\infty$ will be given by

\be
G(\gamma)\equiv\int_{[0,1]^2}G(t,u)\gamma(u,t)dudt.
\ee

The next result lower bounds the distance between the sets ${\cal G}^d$ and ${\cal G}^{\infty}$ via a simple witness.

\begin{prop}
\label{poly_sep}
Consider the symmetric kernel $G(t,u)\equiv 3tu$. Then, 
\be
K^d(G)\equiv\frac{\max\{|G(\gamma')|:\gamma'\in \G^d,\int_{0}^1\gamma'(t,t)dt=1\}}{\max\{|G(\gamma)|:\gamma\in \G,\int_{0}^1\gamma(t,t)dt=1\}}\leq 1-\frac{1}{[d(d-1)+2]^2}.
\label{bound_asymp}
\ee

\end{prop}

This proposition has to be understood as a proof that, in the limit $n\gg d\gg 1$, the difference between $\G^d_n$ and $\G^{\infty}_n$ decreases at least as the inverse of a quartic polynomial in $d$, i.e., it is not exponentially small. Testing free outcome structure microscopic models is therefore potentially practical.

\begin{proof}
In quantum mechanics terminology, $G$ can be seen as a normalized pure state $\tilde{G}=\proj{\psi}$, with $\ket{\psi}=\sqrt{3} \int_0^1 dt t\ket{t}$, with $\braket{t}{u}=\delta(t-u)$. Likewise, any normalized positive semidefinite kernel $\gamma(t,u)$ can be regarded as a normalized a quantum state $\tilde{\gamma}$, and so

\be
G(\gamma)=\tr(\tilde{G}\tilde{\gamma}).
\ee

\noindent If we are optimizing the above quantity over general quantum states $\tilde{\gamma}$, it is clear that the maximum will be attained by taking $\tilde{\gamma}=\tilde{G}$, in which case $G(\gamma)=1$.

Let us now consider the optimization over ${\cal G}^d$. From equation (\ref{contin_d}), it is clear that normalized elements of ${\cal G}^d$ are convex combinations of normalized pure states $\proj{\phi}$ of the form 

\be
\ket{\phi}=\int_0^1dt\sum_{k=1}^{d}\lambda_k(\chi^A_k(t)-\chi^B_k(t))\ket{t}.
\ee

\noindent In turn, such states belong to the bigger set

\be
\Phi\equiv\{\phi:\ket{\phi}=\int_0^1dt\sum_{k=1}^{d(d-1)/2+1}\lambda_k(\chi^A_k(t)-\chi^{B}_k(t))\ket{t},\braket{\phi}{\phi}=1\},
\ee

\noindent where $\lambda_k\geq 0$, and, this time, \emph{all} the sets $\{A_k,B_j\}$ are mutually disjoint, with $\cup_k(A_k\cup B_k)=[0,1]$.

An upper bound on the maximal value of $G(\gamma')$ over the set ${\cal G}^d$ can hence be obtained by solving the optimization problem:

\be
\max_{\phi\in\Phi}\tr(\tilde{\gamma}\proj{\phi})=\max_{\phi\in\Phi}3\left[\int_0^1dt\phi(t)t\right]^2.
\ee

\noindent Now, since $t\geq 0$ in the interval $[0,1]$, we can take all the sets $\{B_k\}$ to be empty. The problem we aim at solving is thus equivalent to finding the best approximation of the function $t$ by a linear combination of $\bar{d}\equiv d(d-1)/2+1$ characteristic functions of mutually disjoint sets. Since the function $t$ is strictly increasing, it is easy to see that the optimal sets $A_k$ must be of the form $A_k=[a_k,a_{k+1}]$, where $a_1=0$ and $a_{\bar{d}+1}=1$.

Given the intervals $[a_k,a_{k+1}]$, the problem of optimizing $\{\lambda_k\}$ can be solved by Fourier analysis. Indeed, it is clear that, for any vector $w$ and any linear subspace $V$, $\max_{v\in V, \|v\|\leq 1} \braket{w}{v}=\sqrt{\bra{w}P_V\ket{w}}$, where $P_V$ is the projector onto the subspace $V$. Since, by definition of the sets $\{A_k\}$, the functions $\{\chi^A_k(t)\}_{k=1}^{\bar{d}}$ are orthogonal, we have that

\be
\bra{\psi}P_{\vec{A}}\ket{\psi}=\sum_{k=1}^{\bar{d}}c_k^2(a_{k+1}-a_k),
\ee

\noindent where 

\be
c_k=\frac{\sqrt{3}\int^1_0 t\chi^A_k(t)dt}{\int^1_0 (\chi^A_k(t))^2dt}=\sqrt{3}\frac{a^2_{k+1}-a_k^2}{2(a_{k+1}-a_k)}.
\ee

\noindent The problem thus reduces to optimize the intervals $[a_k,a_{k+1}]$ so as to maximize $3f(\vec{a})$, where

\be
f(\vec{a})\equiv\sum_{k=1}^{\bar{d}}\frac{(a_{k+1}+a_k)^2(a_{k+1}-a_k)}{4}.
\ee

\noindent Finally, note that

\begin{eqnarray}
&\frac{1}{3}-f(\vec{a})=\frac{1}{3}\sum_{k=1}^{\bar{d}}(a_{k+1}^3-a_k^3)-f(\vec{a})=\nonumber\\
&=\sum_{k=1}^{\bar{d}}\left[\frac{a_{k+1}^3-a_k^3}{3(a_{k+1}-a_k)}-\frac{(a_{k+1}+a_k)^2}{4}\right](a_{k+1}-a_k)=\sum_{k=1}^{\bar{d}}\frac{(a_{k+1}-a_k)^3}{12}.
\end{eqnarray}

\noindent Renaming $v_k=a_{k+1}-a_k$, our aim is now to minimize $\|\vec{v}\|_3$ over all vectors $\vec{v}\in \R^{\bar{d}}$ such that $\|\vec{v}\|_1=1$. This is a well-known problem in Banach space theory: the optimal vector must satisfy $|v_k|=\frac{1}{\bar{d}}$. Hence we end up with:

\be
\max\{\tr(\gamma\gamma'):\gamma'\in \G^d_{\infty}, \int_{0}^1dt\gamma'(t,t)=1\}\leq\min_{\vec{A}}\bra{\psi}P_{\vec{A}}\ket{\psi}=1-\frac{1}{4\bar{d}^2}.
\ee

\noindent Substituting $\bar{d}=d(d-1)/2+1$, we arrive at (\ref{bound_asymp}).

\end{proof}

Note that we can reinterpret eq. (\ref{bound_asymp}) as a lower bound on the trace distance between a certain normalized element of $\G^\infty$ and the set of normalized elements of $\G^d$. Indeed, as we saw during the proof of proposition \ref{poly_sep}, $G\in \G^\infty$ and $\int_0^1G(t,t)t=1$. On the other hand, for any $\gamma\in \G^d$ satisfying the normalization condition, we have that

\be
\|G-\gamma\|_1\geq 2[1-G(\gamma)]\geq \frac{2}{[d(d-1)+2]^2},
\ee

\noindent where the first inequality derives from the relation

\be
\|G-\gamma\|_1=\max_{-\id\leq S\leq \id}\tr\{(G-\gamma)S\}\geq \tr\{(G-\gamma)(2G-\id)\},
\ee
 
\noindent and the second inequality follows from eq. (\ref{bound_asymp}).

\subsection{General algorithm}

Now that we have proven that the general problem makes sense, the next question to ask ourselves is how to solve it in general. That is, given an $n$-variate gaussian probability distribution with covariance matrix $\gamma$, how can we determine if it can be generated through independent $n$-tuples of microscopic variables $(x^l_1,x^l_2,...,x^l_n)\in\{\lambda^l_1,...,\lambda^l_d\}^n$? 

A blind application of Theorem \ref{naked} leads one to the following semidefinite program \cite{sdp}:

\begin{eqnarray}
&&\hspace{30pt}\min \mbox{ }0\nonumber\\
\mbox{s.t.}&&\nonumber\\
&&\gamma_{ij}=\sum_{\vec{c},\vec{c}'}\sum_{a,b}\Gamma(\vec{c},\vec{c}')_{ij}^{ab}\cdot Z_{ab}(\vec{c},\vec{c}'),\nonumber\\
&&Z(\vec{c},\vec{c}')\geq 0.
\label{cosa}
\end{eqnarray}

Indeed, one can see that the covariance matrix $\gamma$ generated by an arbitrary microscopic distribution $p(a_1,a_2,...)$ is equal to

\begin{eqnarray}
&\gamma_{ij}=\sum_{a,b}[p_{ij}(a,b)-p_i(a)p_j(b)]\lambda_a\lambda_b=\sum_{a,b}\sum_{\vec{c},\vec{c}'}\mu_{\vec{c},\vec{c}'}\Gamma_{ij}^{ab}(\vec{c},\vec{c}')\lambda_a\lambda_b=\nonumber\\
&=\sum_{a,b}\sum_{\vec{c},\vec{c}'}\mu_{\vec{c},\vec{c}'}\Gamma_{ij}^{ab}(\vec{c},\vec{c}')\lambda_a\lambda_b=\nonumber\\
&=\sum_{\vec{c},\vec{c}'}\sum_{a,b}\Gamma_{ij}^{ab}(\vec{c},\vec{c}')Z_{ab}(\vec{c},\vec{c}'),
\end{eqnarray}

\noindent with $Z_{ab}(\vec{c},\vec{c}')\equiv \mu_{\vec{c},\vec{c}'}\lambda_a\lambda_b$. By convexity, we thus have that program (\ref{cosa}) defines a set of matrices that contains ${\cal G}^d_n$.

Conversely, let $\gamma=\sum_{a,b}\Gamma_{ij}^{ab}\cdot Z_{ab}$, where $\Gamma$ is a naked covariance matrix and $Z\geq 0$. Then one can Gram-decompose $Z$ as $Z_{ab}=\sum_{i=1}^d\lambda^i_a\lambda^i_b$, and so $\gamma$ is generated by the $d$ $n$-tuples $(x_1^i,...,x_n^i)\in\{\lambda^i_1,...,\lambda^i_d\}^n$, $i=1,...,d$.

\noindent It follows that program (\ref{cosa}) completely characterizes the set ${\cal G}^d_n$.

Program (\ref{cosa}) is, though, unnecessarily inefficient. Notice, for instance, that the covariance matrices do not change if we fix $\lambda_1=0$. Consequently, the matrices $Z(\vec{c},\vec{c})$ can be taken of size $(d-1)\times (d-1)$. Notice as well that certain pairs of points $(\vec{c},\vec{c}')$ and $(\vec{d},\vec{d}')$ actually codify the same information: for example, the matrices $\vec{w}(\vec{c},\vec{c}')\vec{w}(\vec{c},\vec{c}')^T$ and $\vec{w}(\vec{c}',\vec{c})\vec{w}(\vec{c}',\vec{c})^T$ are identical. Also, in the case $d=n=3$ the vectors

\be
w_1=\lambda_2-\lambda_1,w_2=\lambda_1-\lambda_3,w_3=\lambda_2-\lambda_3,
\ee

\noindent and

\be
w_1=\lambda_2-\lambda_1,w_2=\lambda_3-\lambda_2,w_3=\lambda_3-\lambda_1,
\ee

\noindent although seemingly very different, actually contain the same information (namely, that $w_1$ and $w_2$ are free and $w_3=w_1+w_2$). Likewise, the matrices generated by a vector of the form

\be
w_1=\lambda_2-\lambda_1,w_2=0,w_3=\lambda_2-\lambda_1
\ee

\noindent are among the ones generated by any of the other two vectors.

How to avoid this redundancy? The solution is to divide the set of possible points $\{1,...,d\}^n\times\{1,...,d\}^n$ into different classes and choose just one representative (if any) of each class for the semidefinite program. The classification we propose is based on the following fact: given $\vec{c},\vec{c}'$, one can always find a unique parametric representation for $\vec{w}$ of the form $\vec{w}=P\vec{t}$, with $P$ being an $n\times (d-1)$ matrix with entries in $\{-1,0,1\}$ and satisfying the properties:

\begin{enumerate}
\item The first non-zero row of $P$ is $(1,0,...,0)$.
\item Let $N(j)$ denote the greatest from left to the right non-zero index of rows $1,2,...,j$. Then, $N(j+1)\leq N(j)+1$.
\item If $N(j+1)= N(j)+1$, then $P_{j+1,k}=\delta_{k,N(j+1)}$.
\end{enumerate}

\noindent The intuition behind this canonical form is that those rows $j$ where $N(j)=N(j-1)+1$ are independent parameters, while any other row $k$ is a linear combination of the first $N(k)$ `independent rows'. From the proof of Lemma \ref{technical} it is clear that the entries of $P$ have to be $0$'s and $\pm 1$'s. 

Notice that if for some matrix $P$, $N(n)<d-1$, then the covariance matrices generated by $P$ are a subset of those generated by a modified matrix $P'(P)$ where we have made $d-1-N(n)$ dependent rows independent. In sum, we have only to consider those matrices $P$ satisfying conditions 1-3 and generated by a couple of points $\{\vec{c},\vec{c}'\}$ such that $N(n)=d-1$. We will call ${\cal P}^d_n$ the set of all such matrices. Note that ${\cal P}^{n+1}_n=\{\id\}$.

To clarify these ideas, consider the case $n=d=3$. Then, one can check that the elements of ${\cal P}^3_3$ are:

\be
P_1=\left(\begin{array}{cc}0&0\\1&0\\0&1\end{array}\right),P^i_2=\left(\begin{array}{cc}1&0\\i&0\\0&1\end{array}\right),P_3^{i,j}=\left(\begin{array}{cc}1&0\\ 0&1\\i&j\end{array}\right),
\ee

\noindent with $i,j\in\{0,\pm 1\}$. And, consequently, a $3\times 3$ covariance matrix $\gamma$ belongs to ${\cal G}^3_3$ iff

\be
\gamma=P_1Z_1(P_1)^T+\sum_{i=0,1,-1}P^i_2Z^j_2(P^i_2)^T+\sum_{i,j=0,1,-1}P_3^{i,j}Z_3^{i,j}(P_3^{i,j})^T,
\label{program_3}
\ee

\noindent for some $2\times 2$ positive semidefinite matrices $Z_1,Z_2^j,Z_3^{i,j}$. 

How does this set of covariance matrices look like? Figure (\ref{zonas}) shows a plot of the regions $R^d\equiv\{(\vec{v}_1^T\gamma\vec{v}_1, \vec{v}_2^T\gamma\vec{v}_2):\tr(\gamma)=1,\gamma\in \G^d_3\}$, for $d=2,3,4$, with $\vec{v}^T_1=(1,2,4)$, $\vec{v}^T_2=(4,2,1)$. We used the MATLAB package \emph{YALMIP} \cite{yalmip} in combination with \emph{SeDuMi} \cite{sedumi} to perform the numerical calculations. The sets $\G^3_3$ and $\G^4_3$, although very different from $\G^2_3$, seem quite similar to each other in this two-dimensional slice.

\begin{figure}
  \centering
  \includegraphics[width=10 cm]{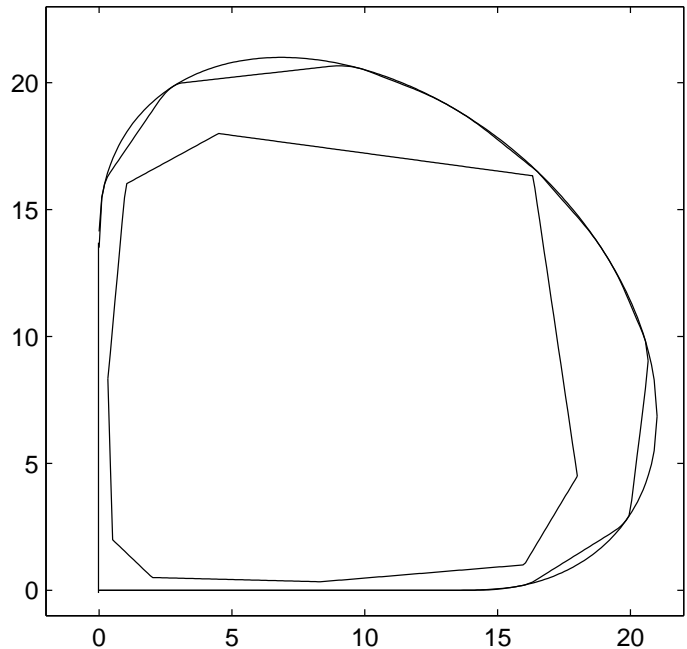}
  \caption{\textbf{Free outcomes.} A slice of the accessible regions of normalized $3\times 3$ covariance matrices generated by tuples of two, three and four-valued microscopic variables.}
  \label{zonas}
\end{figure}

Let us finish with a remark on the applicability of the former results. Notice that program (\ref{cosa}) is only valid if the events that generate the tuple production occur with probablity $\nu=1$ during the course of the experiment. In case $\nu$ is completely unknown, the expression for the covariance matrix $\gamma$ shall be complemented with another summand of the form

\be
\sum_{a,b,\vec{c}}\delta(a,c_i)\delta(b,c_j)Y_{ab}(\vec{c}),
\label{mucho_ruido}
\ee
\noindent with $Y(\vec{c})\geq 0$ for all $\vec{c}\in\{1,...,d\}^n$ (see Section \ref{non_constant}). Moreover, in the regime $\nu\ll 1$, only this summand should appear in $\gamma$'s decomposition. As before, by exploiting redundancies, one can simplify eq. (\ref{mucho_ruido}) to a great extent.

\subsection{The dual problem}

Given that ${\cal G}^d_n$ is a cone, an alternative way to characterize it is through \emph{linear witnesses}, i.e., $n\times n$ matrices $G$ with the property that $\tr(G\gamma)\geq 0$ for all $\gamma\in{\cal G}^d_n$, but such that there exist general covariance matrices $\gamma'\in {\cal G}^{n+1}_n$ with $\tr(G\gamma')< 0$. According to what we have seen, in order to certify that $G$ is indeed a witness for ${\cal G}^d_n$ all we have to do is to check that $G\not\geq 0$ and that $P^TGP\geq 0$ for all $P\in {\cal P}^d_n$. The proof is trivial.

A comparison between the structure of ${\cal G}^d_n$ and those of the sets of classical and quantum correlations is in order. Given a linear functional $\{B^{xy}_{ab}\}$, in order to verify that $\sum_{a,b,x,y}B^{xy}_{ab}P(a,b|x,y)\geq 0$ holds for classical distributions $P(a,b|x,y)$, one has to calculate a finite number of real numbers (the evaluation of $B$ in the extreme points of the classical polytope) and check that all of them are greater or equal than 0. On the other hand, to prove that the inequality holds true for quantum distributions, one would have to show that the (infinitely many) Bell operators that result when we substitute $P(a,b|x,y)$ by projectors of the form $E^x_a\otimes F^y_b$ are positive semidefinite.

In contrast, certifying that $G(\gamma)\geq 0$ for all $\gamma\in {\cal G}^d_n$ amounts to evaluate a finite number of matrices and check that each of them is positive semidefinite. At a purely mathematical level, the set ${\cal G}^d_n$ is thus very interesting, since its behavior is intermediate between the sets of classical and quantum correlations.

\section{Problem complexity}
\label{complex}

Despite our simplifications, the computational power required to run the algorithms presented above to solve the $\G^d_n$-membership problem grows very fast with $n$. This makes one wonder whether it is possible to find simpler schemes to characterize finitely-generated gaussian distributions. In this section we will show a negative result in that direction: deciding whether a given matrix $\gamma$ belongs to ${\cal G}^2_n$ is a {\it strong NP-hard problem}. That is,  there exists a polynomial $q(n)$ such that the membership problem for the set $C_n=\{\gamma\in  {\cal G}^2_n | {\rm tr}(\gamma)=1\}$ cannot be decided efficiently with precision better than $\frac{1}{q(n)}$ in euclidean norm unless $P=NP$.

This result is a trivial consequence of two recent contributions. One, due to Liu \cite[Proposition 2.8]{Liu}, reduces the problem to linear optimization over $C_n$ with inverse polynomial precision. Verifying the hypotheses of Liu's result is trivial in our particular case. Then, one can invoke the following result of Bhaskara et al. \cite[Appendix B]{Bhaskara} 
and notice that QPRATIO is nothing but a linear optimization over $C_n$.
\begin{prop}
There exist two absolute constants $\alpha>\beta$ such that, given a matrix $A$ with coefficients bounded by $1$ and zeros in the diagonal, it is NP-hard to distinguish between a solution $\ge \alpha$ and a solution $\le \beta$ in the following problem, called QPRATIO:

\be
\max_{\{1,-1,0\}^n}\frac{\sum_{i\not = j}A_{ij}x_ix_j}{\sum_ix_i^2}.
\label{QPRATIO}
\ee
\end{prop}

\noindent Indeed, note that eq. (\ref{QPRATIO}) is equal to $\max\{\tr(A\gamma):\gamma\in C_n\}$.

As a final remark, our manuscript shows a natural context  where the maximization problem QPRATIO  appears, which can foster the recently initiated  research on this problem within the computer science community.

\section{Connection with quantum non-locality}
\label{quantum_non_loc}

There exist situations where several macroscopic variables are at stake, but nevertheless the probability density $p(X_1,\ldots,X_n, Y_1,\ldots, Y_n)$ is not experimentally accessible (actually, it may not even exist). This is the case when two parties, call them Alice and Bob, are allowed to interact with their respective ensembles before a measurement is carried out. In such scenarios, for each possible pair of interactions $z,t$, Alice and Bob will be able to estimate the marginals $p(X_z,Y_t|z,t)dX_zdY_t$, see Figure \ref{mac}. Moreover, if Alice and Bob's operations are space-like separated, then $p(X_z|z,t)=p(X_z|z,t')$ ($p(Y_t|z,t)=p(Y_t|z',t)$), for all $t,t'$ ($z,z'$).

Suppose now that all such probability densities happen to be gaussian. A plausible explanation Alice and Bob may come up with is that, whenever they perform an experiment, there is an event in some intermediate region between their labs that produces $N\gg 1$ independent pairs of $d$-leveled particles. By conservation of linear momentum, two particle beams are thus produced, one directed to Alice and another one, to Bob. The action of Alice's (Bob's) interaction $z$ ($t$) affects individually the state of each particle in her (his) beam. After such an interaction, the particles impinge on a simple detector, that produces a state-dependent signal for each incident particle, and Alice and Bob's readings are precisely proportional to such a sum of signals.

\begin{figure}
  \centering
  \includegraphics[width=10 cm]{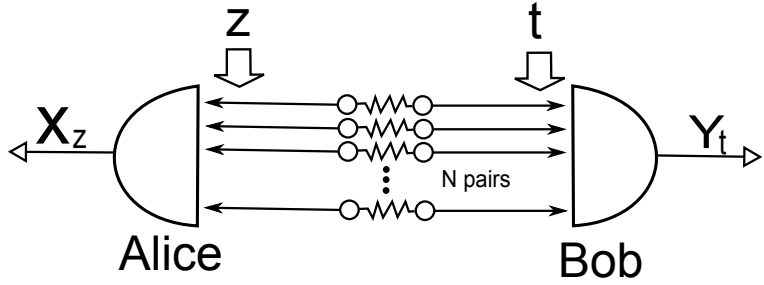}
  \caption{\textbf{Bipartite macroscopic quantum experiment.} Alice and Bob receive each a beam of correlated particles, with which they are allowed to interact in some ways $z$ and $t$ before they reach the detectors. Consequently, they are able to estimate the probability densities $P(X_z,Y_t|z,t)dX_zdY_t$. If their measurements are sufficiently coarse-grained (but precise enough to detect fluctuations in the values of $X_z,Y_t$), this set of gaussian marginal distributions will always admit a local hidden variable model, that is, a classical explanation.}
  \label{mac}
\end{figure}

While the particles in each pair will not be necessarily identical, it is not unreasonable to assume that Alice's (Bob's) interaction does not change the nature of such particles, but only their level probability distribution. That is, even though $\lambda^z_a\not=\lambda^t_a$, it is natural to suppose that $\lambda^z_a=\lambda^{z'}_a$, $\lambda^t_b=\lambda^{t'}_b$. We will introduce the notation $\G^{d_A,d_B}_{n_A,n_B}$ to denote the set of gaussian distributions with $n_A$ ($n_B$) macroscopic variables associated to Alice (Bob), generated by $d_A$($d_B$)-identically-valued classical systems.

Due to Macroscopic Locality (ML) \cite{mac_loc}, for any set of pairs of gaussian distributions 

\be
\{p(X_z,Y_t|z,t)dX_zdY_t\}_{z,t}
\ee

\noindent generated in a quantum bipartite experiment, there is always a global probability distribution $p(X_1,...,X_n,Y_1,...,Y_n)$ that admits $\{p(X_z,Y_t|z,t)dX_zdY_t\}_{z,t}$ as marginals.

One way to interpret this result is that the outcome of any macroscopic quantum experiment can be explained by a classical particle model, where each classical particle reaching Alice's (Bob's) lab produces a current proportional to $X_z$ ($Y_t$), according to the distribution $p(X_1,...,X_n,Y_1,...,Y_n)$. This classical theory, though, seems unnecessarily convoluted: note that, independently of the number of levels $d_q$ in the original quantum experiment, the corresponding classical particles have a continuous spectrum. One thus wonders if there exists a simpler mapping between quantum and classical macroscopic systems that preserves the finite cardinality of the microscopic spectrum. In other words, given a macroscopic quantum experiment with $d_q<\infty$, is it always possible to find a classical description with $d_c<\infty$?


In this section we will answer this question in the negative. More concretely, we will prove that, in order to reproduce quantum macroscopic experiments involving many copies of the maximally entangled two-qubit state ($d_q=2$), one needs to conjure up classical models with particles of infinitely many levels ($d_c=\infty$). So, in a sense, quantum non-locality leaves some signature at the macroscopic level.

Take the following functional, based on the chain inequality \cite{chain}:

\be
G_n(\{\langle \bar{X}_j\bar{Y}_k\rangle\})=\sum_{k=1}^{n}\langle \bar{X}_k\bar{Y}_k\rangle+\sum_{k=1}^{n-1}\langle \bar{X}_{k+1}\bar{Y}_k\rangle-\langle \bar{X}_{1}\bar{Y}_{n}\rangle,
\label{testigo}
\ee

\noindent where $\bar{X}_k\equiv X_k-\langle X_k\rangle, \bar{Y}_k\equiv Y_k-\langle Y_k\rangle$, and consider the problem of minimizing it under the assumption that 

\be
\sum_{k=1}^n\langle\bar{X}^2_k\rangle+\sum_{k=1}^n\langle\bar{Y}^2_k\rangle=1.
\label{norm}
\ee

Then we have the following proposition.

\begin{prop}
\label{quantum}
For any $n$, the minimal value of (\ref{testigo}) compatible with Macroscopic Locality is equal to $g^*_n\equiv -\cos\left(\frac{\pi}{2n}\right)$, and can be attained by performing equatorial measurements over the two-qubit singlet state $\frac{1}{\sqrt{2}}(\ket{0}_A\ket{1}_B-\ket{1}_A\ket{0}_B)$.

\end{prop}
\begin{proof}

From ML we know that, for any set of bipartite marginal distributions, there exists a classical model reproducing the observable second momenta of our macroscopic variables. It follows that optimizing $G_n(\{\langle \bar{X}_j\bar{Y}_k\rangle\})$ against all two-point correlators compatible with ML is equivalent to an optimization over all trace-one positive semidefinite matrices $\gamma$. By convexity, we can thus take $\gamma=\vec{w}(\vec{w})^T$, with $\vec{w}\in\R^{2n}$, $\|\vec{w}\|=1$.

Consequently, minimizing $G_n$ is equivalent to optimize

\be
G_n(\vec{w}^A,\vec{w}^B)=\sum_{k=1}^{n}w^A_kw^B_k+\sum_{k=1}^{n-1}w^A_{k+1}w^B_k-w^A_1w^B_n
\ee

\noindent over all vectors $\vec{w}^A,\vec{w}^B\in\R^n$ such that $\|\vec{w}^A\|^2+\|\vec{w}^B\|^2=1$.

So let us first optimize $\vec{w}^A$. We have that

\begin{eqnarray}
&G(\vec{w}^A,\vec{w}^B)=\sum_{k=2}^nw^A_k\cdot(w^B_k+w^B_{k-1})+w^A_1\cdot(w^B_1-w^B_n)\geq \nonumber\\
&\geq -\sqrt{\sum_{k=2}^n(w^B_{k}+w^B_{k-1})^2+(w^B_{1}-w^B_{n})^2}\|\vec{w}^A\|.
\label{two_val}
\end{eqnarray}

\noindent Our new task is thus to maximize

\be
G'(\vec{w}^B)\equiv\sum_{k=2}^n(w^B_{k}+w^B_{k-1})^2+(w^B_{1}-w^B_{n})^2,
\label{qcirculant}
\ee

Note that we can express $G'(\vec{w}^B)$ as $(\vec{w}^B)^TG'\vec{w}^B$, where $G'$ is an $n\times n$ matrix whose non-zero components are
\begin{eqnarray}
&&G'_{i,i}=2,\mbox{ for } i=1,...,n\nonumber\\
&&G'_{i,i+1}=G'_{i+1,i}=1,\mbox{ for } i=1,...,n-1\nonumber\\
&&G'_{1,n}=G'_{n,1}=-1.
\end{eqnarray}

Gathering intuition from circulant matrices, we try the ansatz $w^B_k=e^{i\alpha k}$ for the eigenvectors of $G'$. One finds that

\be
(G'\vec{w}^B)_k=2e^{i\alpha k}+e^{i\alpha {k+1}}+e^{i\alpha {k-1}}=2[1+\cos(\alpha)]w^B_k,
\ee

\noindent for $k=2,...,n-1$. However, 

\begin{eqnarray}
&&(G'\vec{w}^B)_1=2e^{i\alpha}-e^{i\alpha n}+e^{2i\alpha}. \nonumber\\
&&(G'\vec{w}^B)_n=2e^{i\alpha n}+e^{i\alpha (n-1)}-e^{i\alpha}.
\end{eqnarray}

We circumvent this problem by imposing that $e^{i\alpha(n+1)}=-e^{i\alpha}$. This leads to $\alpha =(2l+1)\pi/n$, with $l=0,1,...,n-1$. This ansatz thus allows us to recover all eigenvectors of $G'$. The greatest eigenvalue is $2(1+\cos(\pi/n))$, and it is attained by the vectors $\{w_k\propto e^{i\pi k/n}\}$ and $\{w_k\propto e^{-i\pi k/n}\}$, or, equivalently, by any vector of the form $\{w_k=\sqrt{2}\sin(k\pi/n+\delta)\}$. The space generated by such vectors will be denoted by $W$, and will play an important role in the proof of Proposition \ref{impossibility}. 

Coming back to our present problem, we have that the minimum value of (\ref{testigo}) upon no restrictions on the gaussian distribution is equal to

\be
-\sqrt{2\left[1+\cos\left(\frac{\pi}{n}\right)\right]}\max_{\|\vec{w}^A\|^2+\|\vec{w}^B\|^2=1}\|\vec{w}^A\|\|\vec{w}^B\|=-\cos\left(\frac{\pi}{2n}\right).
\ee

To complete the proof we thus have to show that the above value can be attained in a macroscopic experiment where many particle pairs in the singlet state are produced and Alice and Bob's interactions correspond to projective equatorial measurements of such particles.

It is straightforward that the measurements

\be
\hat{x}_k=\cos(\psi_k)\sigma_x+\sin(\psi_k)\sigma_z, \hat{y}_k=\cos(\phi_k)\sigma_x+\sin(\phi_k)\sigma_z,
\ee

\noindent with $\psi_k=\frac{(k-1)\pi}{n}$, $\phi_k=\frac{(2k-1)\pi}{2n}$ on both sides of a singlet state saturate the bound $g^*_n$, provided that we assign the values $\{\frac{1}{2n},-\frac{1}{2n}\}$ to each possible outcome.

\end{proof}

We just saw that bi-valued quantum variables suffice to attain the minimum value $g^*_n$ of $G_n(\{\langle \bar{X}_j\bar{Y}_k\rangle\})$. How does this number relate to the minimal value attainable through classical systems of dimensions $d_A,d_B$? The following proposition studies the limits of the $d_A=\infty, d_B=2$-scenario.

\begin{prop}
The minimal value of $\{\tr(G_n\gamma):\tr(\gamma)=1,\gamma\in \G^{\infty,2}_{n,n}\}$ is

\be
g_n^2\equiv-\sqrt{1-\frac{1}{2(n-1)}}.
\ee

\end{prop}

\noindent Note that $g^*_n<g^2_n$ for $n> 2$. Moreover, $\lim_{n\to\infty}\frac{1+g^*_n}{1+g^2_n}=0$.

\begin{proof}
Following the proof of Proposition \ref{quantum}, we just have to maximize (\ref{qcirculant}) for $\vec{w}^B\in \{0,\pm\Delta\}^n$ as a function of $\|\vec{w}^B\|$. Note that the minus sign in (\ref{qcirculant}) between entries $1$ and $n$ destroys translational invariance. However, we can displace it to any position $k$ of the chain by performing the change of coordinates $w^B_i\to -w^B_i$, for $i\leq k$. Now, suppose that $\vec{w}^B\in\{\pm\Delta,0\}^n$ has exactly $s\geq 1$ null entries, choose the first site $k$ such that $w^B_k=0$ and $w^B_{k-1}\not=0$, and place the minus sign between $k$ and $k-1$. That way, for the purposes of maximizing $G'(\vec{w}^B)$, we can always choose the non-zero components equal to $\Delta$. 

Imagine now that all sites except site $k$ are equal to $\Delta$, giving a value $G'(\vec{w}^B)=4\Delta^2(n-1)$. We are now going to place the remaining $s-1$ zeros in the chain starting from site $k$ and in increasing order in such a way as to maximize the value of $G'(\vec{w}^B)$. While `nullifying' each site $i$, we find two different situations: 

\begin{enumerate}

\item The sites $(i-1,i,i+1)$ are of the form $0-\Delta-\Delta$ and we substitute the middle one by a 0, i.e., we will have $0-0-\Delta$. In that case, the value of $G'(\vec{w}^B)$ will decrease by an amount of $4\Delta^2$.
\item The sites $(i-1,i,i+1)$ are of the form $\Delta-\Delta-\Delta$, we substitute the middle one by a 0. That way, we arrive at $\Delta-0-\Delta$, and the corresponding decrease in $G'(\vec{w}^B)$ is equal to $6\Delta^2$.

\end{enumerate}

It follows that, in order to maximize $G'(\Delta)$ for a fixed number of zeros $s$, with $n> s\geq 1$, the best strategy is to place them all one after the other. The maximal value of $G'(\vec{w}^B)$ is thus $(4-\frac{2}{n-s})\|\vec{w}^B\|^2$. The best situation is hence $s=1$, in which case $G'(\vec{w}^B)=(4-\frac{2}{n-1})\|\vec{w}^B\|^2$. It only remains to consider the case $s=0$. But it is immediate that then the maximum value of $G'(\vec{w}^B)$ equals the slightly smaller (for $n\geq 2$) value $(4-8/n)\|\vec{w}^B\|^2$. 

The smallest value of (\ref{testigo}) under the assumption that Bob's microscopic variables are classical and bivalued is therefore 

\be
\min_{\|\vec{w}^A\|^2+\|\vec{w}^B\|^2=1}-\sqrt{\max_{\|\vec{u}\|=1} G'(\vec{u})}\|\vec{w}^A\|\|\vec{w}^B\|=-\sqrt{1-\frac{1}{2(n-1)}}.
\ee

\end{proof}

The case $d_A=\infty, d_B>2$ is much more difficult to analyze, and so this time we can only provide gross lower bounds $h^d_n$ for $g^d_n\equiv\min\{\tr(G_n\gamma):\tr(\gamma)=1,\gamma\in\G^{\infty,d}_{n,n}\}$. Such bounds, though, satisfy $h^d_n>g^*_n$ for $d\sim O(\sqrt{n})$, thereby proving that no finitely-valued classical model can account for the macroscopic correlations we observe when we perform polarization measurements over different ends of a double beam of many photons in the maximally entangled state.

\begin{prop}
\label{impossibility}
Let $\frac{1}{2}d(d-1)+3\leq \lceil \frac{n}{2}\rceil$. Then, $g^d_n>g^*_n$.
\end{prop}

\begin{proof}
From the proof of Proposition \ref{quantum}, we have that, in order to attain $g^*_n$, Bob's vector has to be parallel to a vector in $W$, the 2-dimensional space of eigenvectors with minimum eigenvalue. Hence, if we manage to prove that $\vec{w}^B$ cannot approximate any such vector, we are done. 

Suppose then that $\frac{1}{2}d(d-1)+3\leq\lceil\frac{n}{2}\rceil$, and call $\vec{u}$ the (normalized) projection of $\vec{w}^B$ onto $W$. Then, the number of different non-negative (or non-positive) entries $\{\lambda_i-\lambda_j:i,j=1,...,d\}$ of $\vec{w}^B$ is at most $\lceil\frac{n}{2}\rceil-2$. Now, any unitary vector $\vec{u}\in W$ is of the form $u_k=\sqrt{\frac{2}{n}}\sin\left(\frac{\pi}{n}k+\delta'\right)$. Identifying the $(n+1)^{th}$ entry with the first entry, we thus have that $\vec{u}$ must have a sequence of consecutive entries of the form $\pm \frac{\sqrt{2}}{n}\sin\left(\frac{\pi}{n} j+\delta\right)$, for $j=0...\lceil\frac{n}{2}\rceil-2$, with $0\leq\delta<\frac{\pi}{n}$. It follows that there exist at least two entries $k,l$ such that $w^B_k=w^B_l$ and $u_k=\pm\sqrt{\frac{2}{n}}\sin(\pi m_k/n+\delta)$, $u_l=\pm\sqrt{\frac{2}{n}}\sin(\pi m_k/n+\delta)$, with $m_k,m_l\in\N, m_k\not=m_l$. The minimum of the quantity $\|\vec{w}^B-\vec{u}\|^2$ is thus lower bounded by

\begin{eqnarray}
&\min_{\vec{u}\in W}\|\vec{w}^B-\vec{u}\|^2\geq (w^B_l-u_l)^2+(w^B_l-u_k)^2\geq \frac{1}{2}(u_l-u_k)^2\geq\nonumber\\
&\geq \min_{m_k\not= m_l}\frac{1}{n}[\sin\left(m_k\frac{\pi}{n}+\delta\right)-\sin\left(m_l\frac{\pi}{n}+\delta\right)]^2\geq \frac{4}{n}\sin^2\left(\frac{\pi}{2n}\right)\sin^2\left(\frac{\pi}{n}\right)>0,
\end{eqnarray}

\noindent where the last inequality follows from taking $m_k=\lceil\frac{n}{2}\rceil-2,m_l=\lceil\frac{n}{2}\rceil-3$ and applying some trigonometric identities.

It follows that $\vec{w}^B\not\in W$, and so the minimum $g^*_n$ cannot be attained. It is straightforward to carry on this argument in order to find a horribly complicated lower bound on the difference $g^d_n-g^*_n$.

\end{proof}

In view of this last witness $G_n$, whose maximal value was attained by quantum mechanical bivalued variables, one may be tempted to think that quantum dichotomic systems suffice to reproduce all gaussian marginal distributions arising in a bipartite scenario. However, such is not the case. Consider a scenario where Alice has only one interaction to choose from. Then we have the following result.

\begin{prop}
Let Alice and Bob's macroscopic observables arise as sums of microscopic variables subject to the no-signalling constraint only, with $d_A=2$ and $d_B=n$. Then, the macroscopic variables satisfy: 

\be
\frac{1}{2}\sqrt{\frac{4^n-1}{3}-\frac{1}{n}}\left[\langle\bar{X}\rangle^2+\sum_{k=1}^n\langle\bar{Y}_k\rangle^2\right]-\sum_{k=1}^n2^{k-1}\langle\bar{X}\bar{Y}_k\rangle\geq 0.
\label{hope}
\ee

\end{prop}

\noindent This inequality can be violated by the classical matrix 

\be
\gamma=\left(\begin{array}{c}w^A\\\vec{w}^B\end{array}\right)\left(\begin{array}{c}w^A\\\vec{w}^B\end{array}\right)^T,
\ee

\noindent with $w^A=1$, $w^B_k=\sqrt{\frac{3}{4^n-1}}2^{k-1}$, for $k=1,...,n$.

\begin{proof}
Since Alice has only one interaction available, there exists a global probability distribution for the variables $(x,y_1,...,y_n)$ (take, for instance, $p(a,b_1,...,b_n)\equiv\frac{1}{p(a)^{n-1}}\prod_{i=1}^np(a,b_i)$). It follows that we can regard the scenario as classical and apply the techniques developed so far. The problem thus reduces to minimizing the expression

\be
(\vec{v}\cdot\vec{w}^B)w^A,
\ee

\noindent where $v_i=2^{i-1},i=1,...,n$, for a fixed value of $\|\vec{w}^B\|^2+(w^A)^2$. Following the proof of Proposition \ref{identity}, we have that $\vec{w}^B\perp \vec{s}$, with $s_i=\chi_C(i)-\chi_D(i)$, for two disjoint sets of outcomes $C$ and $D$ such that $C\cup D\not=\emptyset$. Using the same arguments, one concludes that $|\vec{v}\cdot \vec{w}^B|\leq \sqrt{\|\vec{v}\|^2-\frac{1}{n}}\|\vec{w^B}\|$, and so we arrive at (\ref{hope}).

\end{proof}

\section{Conclusion}
\label{conclusion}

In this paper we have shown that different microscopic spectra lead to different restrictions on the macroscopic probability distributions that result out of the convolution of a big number of independent microscopic systems. To put it in another way: it is possible to extract non-trivial information about the output structure of a microscopic system just by analyzing the corresponding gaussian macroscopic variables. We have shown that characterizing the set of all gaussian distributions generated by microscopic systems with a given (finite) spectrum can be formulated as a linear program. We have also proven the existence of infinite dimensional outcome structures that do not allow to recover the set of all gaussian bivariate distributions. In the case where the spectrum of the microscopic variables is not fixed a priori, we have shown that $n+1$ outcomes are enough to recover all $n$-variate gaussian distributions, and this number is optimal. We have demonstrated that the problem of characterizing gaussian distributions generated by independent tuples of $d$-valued microscopic variables with free outcome structure can be cast as a semidefinite program. Furthermore, we studied the complexity class of characterizing the witnesses of the set of gaussians generated by dichotomic variables. Despite the fundamental nature of all these problems, to our knowledge, this line of research had never been previously considered in Probability Theory.

We used the above results to prove that bipartite gaussian states of light cannot be regarded as ensembles of pairs of yet-to-be-discovered $d$-level particles. More concretely, we showed that, for any value of $d$, there exists a quantum optics experiment that proves that gaussian states do not follow the microscopic $(d-1)/2$-spin model. Likewise, we saw that classical models aiming at describing certain macroscopic quantum experiments involving dichotomic quantum systems necessarily involve infinitely-valued classical variables. From a foundational point of view, these two results justify the whole study. We believe, though, that the range of applicability of the techniques developed in this work does not end here, and that they will eventually find a use in other fields.

From a strictly mathematical point of view, this work leaves open some problems that we believe are important to address:

\begin{enumerate}

\item
In Section \ref{quantum_non_loc}, we showed a witness for $\G^{\infty,O(\sqrt{n})}_{n,n}$ that was violated by quantum dichotomic systems. However, the exact value $g^d_n\equiv\min\{\tr(G_n\gamma):\tr(\gamma)=1,\gamma\in \G^{\infty,d}_{n,n}\}$ could not be computed exactly. Is it possible to improve the bounds for $g^d_n$? Alternatively, is there a better witness to separate the quantum bivalued set from classical many-valued models? Note that $G_n$ is not very robust against detector noise.

\item
If very high precision measurements are available, we may obtain (non-zero) estimates of higher order cumulants of $X_1,X_2,...,X_n$. Can we use these estimates to extract more information about $x_1,...,x_n$?

\item
In this work, we have characterized the set of gaussian distributions arising from classical finitely-valued systems. Can our results be extended to the quantum case? That is, given a set of marginal gaussian distributions $P(X_z,Y_t|z,t)$, is it possible to determine if they can be generated by microscopic quantum systems? In the fixed outcome case, it is easy to think of a hierarchy of SDPs to \emph{bound} such a set of gaussian distributions: following the derivation of Theorem \ref{naked}, consider four-partite quantum distributions $p(a,b,a',b'|z,t,z',t')$ which are invariant with respect to the interchange $(a,b)\leftrightarrow (a',b')$. Such distributions can be bounded by sequences of SDPs \cite{hier1,hier2,hier3}, and the cone of naked covariance matrices generated by $p(a,b|z,t)$ can in turn be bounded by a linear expression on $p(a,b,a',b'|z,t,z',t')$. Does this approach converge to the desired set of gaussian distributions? If not, one could consider extensions of $p(a,b|z,t)$ to $2n$ parties and invoke some (to be discovered) de Finetti theorem for quantum correlations. Does this new approach converge? Even better: is it possible to characterize `quantum' covariance matrices with just \emph{one} SDP?

\end{enumerate}

\section*{Acknowledgements}
MN acknowledges support by the Templeton Foundation and the European Commission (Integrated Project QESSENCE). This work was supported by the Spanish grants I-MATH, MTM2011-26912, QUITEMAD and the European project QUEVADIS. We acknowledge Oliver Johnson for discussions about the results of this paper.

\end{document}